\pdfoutput=1
\documentclass{llncs}
\usepackage{llncsdoc}
\usepackage{tabularx,colortbl}
\usepackage[table]{xcolor}
\usepackage[pdftex]{graphicx}
\usepackage{subfigure}
\usepackage{amsmath}
\usepackage{amssymb}
\usepackage{mathrsfs}
\usepackage{appendix}
\usepackage{algorithm}
\usepackage{algpseudocode}
\usepackage{xypic}
\usepackage{subfigure}
\usepackage{enumerate}
\usepackage{pifont}  
\usepackage{setspace}
\usepackage{mathtools}

\usepackage{ifthen} 
\usepackage{xifthen} 

\usepackage{verbatim}

\begin{document}

\hyphenation{con-cu-rrency}
\hyphenation{cu-rrent}
\hyphenation{co-ro-llary}
\hyphenation{rea-cha-ble}
\hyphenation{de-ci-da-ble}
\hyphenation{mann-er}
\hyphenation{cha-nnel}
\hyphenation{cha-rac-te-ri-stic}
\hyphenation{bei-ng}
\hyphenation{fo-llow-ing}
\hyphenation{be-ginn-ing}
\hyphenation{cha-llen-ger}
\hyphenation{de-fen-der}
\hyphenation{co-rre-sponding}
\hyphenation{di-ffe-rence}
\hyphenation{con-si-der}
\hyphenation{con-si-dered}
\hyphenation{bi-simu-la-tion}
\hyphenation{bi-simi-la-rity}
\hyphenation{bi-simi-lar}
\hyphenation{pre-bi-simu-la-tion}
\hyphenation{pre-bi-simi-la-rity}
\hyphenation{bet-ween}
\hyphenation{tran-si-tion}
\hyphenation{lab-elled}
\hyphenation{auto-ma-ta}
\hyphenation{auto-ma-ton}
\hyphenation{syn-the-size}
\hyphenation{charac-te-ri-za-tion}
\hyphenation{charac-te-ri-za-tions}
\hyphenation{equi-va-lence}
\hyphenation{ob-ser-va-tional}
\hyphenation{assign-ing}
\hyphenation{pre-or-der}
\hyphenation{ana-ly-sis}
\hyphenation{ana-ly-zing}
\hyphenation{modell-ing}
\hyphenation{sy-mmet-ric}
\newcommand{\track}[1]{{\textcolor{red}{#1}}}
\newcommand{\remove}[1]{{\textcolor{green}{#1}}}
\newcommand{\delete}[1]{{}}
\newcommand{\sembrack}[1]{[\![#1]\!]}
\newcommand{\existsdelay}{\exists \!\!\!\! \exists}
\newcommand{\foralldelay}{\forall \!\!\!\! \forall}
\newcommand{\realpos}{\mathbb{R}_{\ge 0}}
\newcommand{\Tau}{\mathrm{T}}
\newcommand{\cmark}{\ding{52}}
\renewcommand{\to}[1]{\xrightarrow{#1}}
\newcommand{\isin} {\:\mbox{\texttt{\underline{in}}}\:}
\newcommand{\setg}[2]{\Upsilon_{(#1,#2)}}
\newcommand{\starts}[2]{\mathrm{Strt}(#1)_{#2}}
\newcommand{\send}[2]{\mathrm{End}(#1)_{#2}}
\newcommand{\range}[2]{\mathrm{Ran}(#1)_{#2}}
\newcommand{\lowerbound}[1]{LB(#1)}
\newcommand{\upperbound}[1]{UB(#1)}
\newcommand{\beforeg}{\lll}
\newcommand{\before}{\lessdot}
\newcommand{\set}[1]{\{{#1}\}}
\renewcommand{\diamond}[1]{\langle #1 \rangle}
\DeclarePairedDelimiter\ceil{\lceil}{\rceil}
\DeclarePairedDelimiter\floor{\lfloor}{\rfloor}
\newcommand{\mlineqn}[5]{\begin{center}\ensuremath{#1 = \left\{ 
  \begin{array}{l l}
    #2 & \quad \text{#3}\\
    #4 & \quad \text{#5}
  \end{array} \right.}\end{center}}

\newcommand{\TGame}[3]{%
\ifthenelse{\isempty{#1}}
{\ifthenelse{\isempty{#3}}{\Gamma_{#2}^{{#3}}}{\Gamma_{#2}^{{#3}}}}               
            {#1\!-\!\Gamma_{#2}^{{#3}}}    
}

\newtheorem{fact}[theorem]{Fact}

\title{Reducing Clocks in Timed Automata while Preserving Bisimulation}

\author{Shibashis Guha\thanks{The research of Shibashis Guha was supported by
Microsoft Corporation and Microsoft Research India under the Microsoft Research India PhD Fellowship Award.}
$\:\;\:\:\;\:\:\;\:$ \quad \qquad Chinmay Narayan $\:\;\:\:\;\:\:\;\:$ \quad \qquad S. Arun-Kumar}

\institute{Indian Institute of Technology Delhi \\
\email{\{shibashis, chinmay, sak\}@cse.iitd.ac.in}
}

\maketitle

\begin{abstract}
Model checking timed automata becomes increasingly complex with the increase in the number of clocks. Hence it is desirable that one constructs an automaton with the minimum number of clocks possible. The problem of checking whether there exists a timed automaton with a smaller number of clocks such that the timed language accepted by the original automaton is preserved is known to be undecidable.
In this paper, we give a construction,
which for any given timed automaton produces a timed bisimilar automaton 
with the least number of clocks. Further, we show that such an automaton with the minimum possible number of clocks can be constructed in time that is doubly exponential in the number of clocks of the original automaton.
\end{abstract}


\section{Introducton} \label{sec-clockredIntro}
Timed automata \cite{AD1} is a formalism for modelling and analyzing real time systems. The complexity of model checking is dependent on the number of clocks of the timed automaton (TA) \cite{AD1,ACD1}. 
Model checkers use the region graph or a zone graph construction
for analysing reachability and other properties in timed automata.
These graphs have sizes exponential in the number of clocks of the timed automaton.
The algorithms for model checking in turn depend on the sizes of these graphs.
Hence it is desirable to construct a timed automaton with the minimum number of clocks that preserves some property of interest
(such as timed language equivalence or timed bisimilarity).
Here we show that checking the existence of a timed automaton with fewer clocks
that is timed bisimilar to the original timed automaton is decidable.
Our method is constructive and we provide a 2-EXPTIME algorithm to construct the timed bisimilar automaton with the least possible number of clocks. 
We also note that if the constructed TA has a smaller number of clocks,
then it implies that there exists an automaton with a smaller number of clocks accepting the same timed language.

\textbf{Related work:} 
In \cite{DY1}, an algorithm has been provided to reduce the number of clocks of a given timed automaton.
It produces a new timed automaton that is timed bisimilar to the original one. The algorithm detects a set of \emph{active clocks} at every location and partitions these active clocks into classes such that all the clocks belonging to a class in the partition always have the same value.
However, this may not result in the minimum possible number of clocks.
The algorithm described in the paper \cite{DY1} works on the timed automaton directly rather than on its semantics.
For instance, if a constraint associated with clock $x$ implies a constraint
associated with clock $y$, and the conjunction of the two constraints appears on an edge,
then the constraint associated with clock $y$
may be eliminated.
However, the algorithm of \cite{DY1} does not capture such implication.
Also by considering 
constraints on more than one outgoing edge from a location,
e.g.  $l_0 \to{a, x \le 3, \emptyset} l_1$ and $l_0 \to{a, x > 3, \emptyset} l_2$
collectively, we may sometimes eliminate the constraints
that may remove some clock. This too has not been accounted for by the algorithm of \cite{DY1}.

In a related line of work
\cite{Trip1}, it has been shown that no algorithm can decide the
following two things together.
\begin{enumerate}[(i)]
\item minimality of the number of clocks while preserving the timed language and 
\item for the non-minimal case finding a timed language equivalent automaton with fewer clocks.
\end{enumerate}
Also as mentioned earlier, for a given timed automaton,
the problem of determining whether there exists another TA with fewer clocks accepting the same timed language is undecidable\index{clock reduction!preserving timed language} \cite{Fin1}.

Another result appearing in \cite{LLW1} which uses the region-graph construction is the following.
A $(C, M)$-automaton\index{(C, M)-automaton} is one with $\lvert C \lvert$ clocks and $M$ is the largest integer appearing in the timed automaton.
Given a timed automaton $A$, a set of clocks $C$ and an integer $M$, checking the existence of a ($C, M$)-automaton that is timed bisimilar to $A$ is shown to be decidable.
The method described there constructs a logical formula called the \emph{characteristic formula} and checks whether there exists a $(C,M)$-automaton that satisfies it.
Further, it is shown that a pair of automata satisfying the same characteristic formula are timed bisimilar.
The paper also considers the following problem $P$ and leaves it open.
\begin{quote}
Does there exist a timed automaton with $|C|$ clocks
that is timed bisimilar to $A$?
\end{quote}
We solve the following problem $Q$.
\begin{quote}
Given a timed automaton $A$, construct a timed bisimilar automaton $B$
with the least number of clocks.
\end{quote}
Solving the problem $Q$ also solves the problem $P$ and vice versa.
If $B$ has say $|C'|$ clocks and $|C| \ge |C'|$,
then the answer to the problem $P$ is in the affirmative.
For the problem $P$, if the answer is `yes' for $|C|$ clocks
but the answer is `no' for $|C|-1$ clocks,
then it is clear that the timed automaton
with the minimum possible number of clocks that is timed bisimilar to $A$
can have no less than $\lvert C \lvert$ clocks.

The rest of the paper is organized as follows: in Section \ref{sec-ta}, we describe timed automata and introduce several concepts that will be used in the paper. We also describe the construction of the zone graph
in Section \ref{sec-zoneGraph} used in reducing the number of clocks.
This construction of the zone graph has also been used in \cite{GKNA1}
for deciding several timed and time abstracted relations.
In Section \ref{sec-method}, we discuss our approach in detail along with a few examples. Section \ref{sec-conc} is the conclusion.

\section{Timed Automata}\label{sec-ta}
Formally, a \emph{timed automaton} (TA) \cite{AD1} is defined as a tuple $A=(L, Act, l_0, C, E)$ where
$L$ is a finite set of locations, $Act$ is a finite set 
of visible actions, $l_0 \in L$ is the initial location,
$C$ is a finite set of clocks and $E \subseteq L \:\times\: \mathcal{B}(C) \:\times\: Act \:\times\: 2^C \:\times\: L$
is a finite set of \emph{edges}. The set of constraints or guards on the edges, denoted $\mathcal{B}(C)$, is 
given by the grammar $g ::= \; x \bowtie k \:|\: g \wedge g$,
where $k \in \mathbb{N}$ and $x \in C$ and $\bowtie \: \in \: \{\le, <, =, >, \ge\}$.
Given two locations $l, l'$, a transition from $l$ to $l'$ is of the form 
$(l, g, a, R, l')$ i.e. a transition from $l$ to $l'$  on action $a$ is possible if the constraints 
specified by $g$ are satisfied; $R \subseteq C$ is a set of clocks which are reset to zero during  
the transition.

The semantics of a timed automaton(TA) is described by a \emph{timed labelled
transition system} (TLTS) \cite{LAKJ1}. The timed labelled transition system $T(A)$ generated by $A$ is defined as 
$T(A) = (Q, Lab, Q_0, \{\stackrel{\alpha}{\longrightarrow} | \alpha \in Lab\})$, where 
$Q \: = \: \{(l,v)\:|\: l \in L, v \in {\realpos}^{|C|}\}$ is the set of  \emph{states}, each of which is of the form $(l,v)$, where $l$ is a location of the timed automaton 
and $v$ is a valuation in $|C|$ dimensional real space
where each clock is mapped to a unique dimension in this space;
$Lab = Act \cup \realpos$ is the set of labels. 
Let $v_0$ denote the valuation such that $v_0(x)=0$ for all 
$x \in C$. $Q_0=(l_0, v_0)$ is the initial state of $T(A)$.
A transition may occur in one of the following ways:\\ 
(i) \emph{Delay transitions} : $(l,v)\stackrel{d}{\longrightarrow}(l, v+d)$. Here, $d \in \realpos$ and $v + d$ is the valuation
in which the value of every clock is incremented by $d$. \\
(ii) \emph{Discrete transitions} : $(l,v) \stackrel{a}{\longrightarrow} (l', v')$ if for an edge 
$e=(l,g,a,R,l') \in \: E$,
$v$ satisfies the constraint $g$, (denoted $v \models g$),
$v' = v_{[R \leftarrow \overline{0}]}$,
where $v_{[R \leftarrow \overline{0}]}$ denotes that every clock in $R$ has been reset to 0, while the remaining clocks are unchanged.
From a state $(l,v)$, if $v \models g$, then there exists an $a$-transition to a state $(l', v')$; 
after this,  the clocks in $R$ are reset while those in $C\backslash R$ remain unchanged. 

For simplicity, we do not consider annotating 
locations with clock constraints (known as \emph{invariant conditions} \cite{HNSY1}). Our results extend in a
straightforward manner to timed automata with invariant conditions.
In Section \ref{sec-method}, we provide the modifications to our method for dealing with location invariants.
We now define various concepts that will be used in the rest of the paper.
\begin{definition} 
Let $A=(L, Act, l_0, E, C)$ be a timed automaton, and $T(A)$ be the TLTS corresponding to $A$. 
\begin{enumerate}
\item \textbf{Timed trace}: A sequence of delays and visible actions $d_1 a_1 d_2 a_2 \dots d_n a_n$
is called a timed trace of the timed automaton $A$ iff there is a sequence of transitions 
$p_0 \to{d_1} p_1 \to{a_1} p_1'\to{d_2}p_2\to{a_2}p_2'\cdots \to{d_n}p_n\to{a_n}p_n'$ in $T(A)$, 
where each $p_i$ and $p_i'$, $0 \le i \le n$
are timed automata states and are of the form $(l, v)$.
$p_0$ is the initial state of the timed automaton.
Given a timed trace $tr = d_1 a_1 d_2 a_2 \dots d_n a_n$,
$untime(tr)$ denotes its projection on $Act$, i.e.
$untime(tr) = a_1 a_2 \dots a_n$.
\item \textbf{Zone and difference bound matrix (DBM)}: A zone $Z$ is a set of valuations $\{v \in \realpos^{|C|} \mid v \models \beta\}$, where 
$\beta$ is of the form $\beta ::= \; x \bowtie k \:|\: x-y \bowtie k \:|\: \beta \wedge \beta$,
$k$ is an integer, $x,y \in C$ and $\bowtie \: \in \: \{\le, <, =, >, \ge\}$.
$Z \!\!\uparrow$ denotes the future of the zone $Z$.
$Z\!\!\uparrow=\{v+d \mid v \in Z, d \geq 0\}$ is the set of all valuations 
reachable from $Z$ by delay transitions.
A zone is a convex set of clock valuations.
A difference bound matrix is a data structure used to
represent zones.
A DBM  \cite{BM1,Dill1} for a set $C = \{x_1, x_2, \dots, x_n\}$ of $n$ clocks is an $(n+1)$ square matrix $\Xi$
where an extra variable $x_0$ is introduced such that the value of $x_0$ is always $0$.
An element $\Xi_{ij}$ is of the form $(m_{ij}, \prec)$ where $\prec \in \{<, \le\}$ 
such that $x_i - x_j \prec m_{ij}$.
Thus an entry $m_{i0}$ denotes the constraint $x_i - x_0 \prec m_{i0}$ which is equivalent to $x_i \prec m_{i0}$.
\item \textbf{Canonical decomposition}: Let $\displaystyle{g =\bigwedge_{i=1}^n \gamma_i \in \mathcal{B}(C)}$, where 
each $\gamma_i$ is an \emph{elementary constraint} of the form $x_i \bowtie k_i$,
such that $x_i \in C$ and $k_i$ is a non-negative integer.
A canonical decomposition of a zone $Z$ with respect to  $g$ 
is obtained by splitting $Z$ into a set of zones $Z_1, \dots, Z_m$ such that 
for each $1 \leq j \leq m$, and $1 \leq i \leq n$, either $\forall v \in Z_j, v \models \gamma_i$ or $\forall v \in Z_j$, $v \not \models \gamma_i$.
For example, consider the zone $Z=x \geq 0 \wedge y \geq 0$ and the guard  $x \le 2 \wedge y > 1$. 
$Z$ is split with respect to $ x \le 2$, and then with respect to $y > 1$, hence 
into four zones : 
$x \le 2 \wedge y \le 1$, $x > 2 \wedge y \le 1$, $x \le 2 \wedge y > 1$ and $x > 2 \wedge y > 1$.
An elementary constraint $x_i \bowtie k_i$ induces the hyperplane $x_i = k_i$ in a zone graph of the timed automaton.
\item \textbf{Zone graph}:
Given a timed automaton $A= (L, Act, l_0, C, E)$, a \emph{zone graph} $\mathcal{G}_A$ of $A$
is a transition system $(S, s_0, Lep, \rightarrow)$, 
that is a finite representation of $T(A)$.
Here $Lep = Act \cup \{\varepsilon\}$.
$S \subseteq L \times \mathcal{Z}$ is the set of nodes of $\mathcal{G}_A$, $\mathcal{Z}$ being the set of zones.
The node $s_0 = (l_0, Z_0)$ is the initial node such that $v_0 \in Z_0$.
$(l_i, Z) \to{a} (l_j, Z')$ iff $l_i \to{g, a, R} l_j$ in $A$ and $Z' \subseteq ([Z \cap g]_{R \leftarrow \overline{0}})\!\!\uparrow$ obtained after canonical decomposition of $([Z \cap g]_{R \leftarrow \overline{0}})$.
For any $Z$ and $Z'$, $(l_i, Z) \to{\varepsilon}(l_i, Z')$ iff there exists a delay $d$ and a valuation $v$ such that $v\in Z$, $v+d \in Z'$ and $(l_i,v)\stackrel{d}{\longrightarrow}(l_i, v+d)$ is in $T(A)$. Here the zone $Z'$ is called a \emph{delay successor} of zone $Z$, while $Z$ is called the \emph{delay predecessor} of $Z'$.
Every node has an $\varepsilon$ transition to itself and the $\varepsilon$ transitions are also transitive.
The relation $\to{\varepsilon}$ is reflexive and transitive and so is the \emph{delay successor} relation.
We denote the set of delay successor zones of $Z$ by $ds(Z)$.
A zone $Z' \neq Z$  is called the \emph{immediate delay successor} of a zone $Z$ iff $Z' \in ds(Z)$ and
$\forall Z'' \in ds(Z): Z'' \neq Z$ and $Z'' \neq Z'$, $Z'' \in ds(Z')$.
We call a zone $Z$ corresponding to a location to be a \emph{base zone} if $Z$ does not have a delay predecessor
other than itself.
For both $a$ and $\varepsilon$ transitions, if $Z$ is a zone then $Z'$ is also a zone, i.e. $Z'$ is a convex set.
\item A hyperplane $x = k$ is said to \emph{bound} a zone $Z$ from above if 
$\exists v \in Z \:[\:\forall d \in \realpos \:[\: (v+d)(x) \succ k \iff (v+d)(x) \notin Z \:]\:]$, where $\succ \in \{>, \ge \}$.
A zone, in general, can be bounded above by several hyperplanes.
A hyperplane $x = k$ is said to \emph{bound} a zone $Z$ fully from above if 
$\forall v \in Z \:[\:\forall d \in \realpos \:[\: (v+d)(x) \succ k \iff (v+d)(x) \notin Z \:]\:]$.
Analogously, we can also say that a hyperplane $x = k$ bounds a zone from below if
$\exists v \in Z \:[\:\forall d \in \realpos \:[\: (v-d)(x) \prec k \iff (v-d)(x) \notin Z \:]\:]$, where $\prec \in \{<, \le \}$.
We can also define a hyperplane bounding a zone fully from below in a similar manner.

When not specified otherwise, in this paper, a hyperplane bounding a zone implies
that it bounds the zone from above.
A zone $Z$ is bounded above if it has an immediate delay successor zone.
\item \textbf{Pre-stability}: A zone $Z_1$ of location $l_1$ is pre-stable with respect to a transition $\eta \in Lep$
leading to a zone $Z_2$ of location $l_2$ if\\
i) for $\eta = a \in Act \: :$ $\forall v_1 \in Z_1, \exists v_2 \in Z_2$ such that $(l_1, v_1) \to{a} (l_2, v_2)$ or
$\forall v_1 \in Z_1, \not\!\exists v_2 \in Z_2$ such that $(l_1, v_1) \to{a} (l_2, v_2)$. \\
ii) for $\eta = \varepsilon \: :$ $\forall v_1 \in Z_1, \exists v_2 \in Z_2, \exists d \in \realpos$ such that $(l_1, v_1) \to{d} (l_2, v_2)$ or
$\forall v_1 \in Z_1$, $\not\!\exists v_2 \in Z_2$ such that $(l_1, v_1) \to{d} (l_2, v_2)$, for any $d \in \realpos$. \\
A zone graph is said to be pre-stable if the zones corresponding to each of the nodes
of the zone graph are pre-stable with respect to their outgoing transitions.
For example zone $Z_1$ in Figure \ref{fig-delayZone} is pre-stable with respect to zone $Z_4$,
however, $Z_1 \cup Z_2 \cup Z_3$ is not pre-stable with respect to zone $Z_4$
since $\exists v \in Z_3$ and hence $v \in Z_1 \cup Z_2 \cup Z_3$ such that for all $d \in \realpos$, $v + d \notin Z_4$.
\item \textbf{corner point and corner point traces}: A corner point\index{corner point} of a zone in the zone graph is a state
where corresponding to each clock $x$, a hyperplane of the form $x = c$ bounding a zone intersects
another hyperplane defining the zone.
Considering $M$ to be the largest constant appearing in the constraints or the location invariants of the timed automaton,
each coordinate of a corner point is of the form $n, n + \delta$ or $n - \delta$ where 
$n \in \{0, 1, \dots, M\}$ and $\delta$ is a symbolic value for an infinitesimally small quantity.
In a zone $Z$, there can be two kinds of corner points, an \emph{entry}\index{corner point!entry} corner point and an \emph{exit}\index{corner point!exit} corner point. For an entry corner point, there is a unique exit corner point, that can be reached from the entry corner point by performing a delay $d$ such that any delay more than $d$ will lead to a state that is in a delay successor zone of $Z$. 

For a zone, it is possible to have a pair of such entry and exit corner points that are not distinct, for example, in Figure \ref{fig-delayZone}, in zone $Z_4$, the corner point ($x = 5 + \delta, y = 7$) is both an entry as well as an exit corner point, while the corner point ($x = 5 + \delta, y = 5+\delta$) is another entry corner point for the same zone $Z_4$ and the corresponding exit corner point is $(x = 7, y = 7)$.
If the entry and the exit corner points are distinct,
then a non-infinitesimally small delay can be performed from the entry corner point
to reach the corresponding exit corner point whereas an infinitesimally small delay from the exit corner point causes it to evolve into an entry corner point of the immediate delay successor zone.
A zone that is not bounded above does not have any exit corner point.

For an entry corner point, the corresponding exit corner point is the \emph{next delay corner point} whereas for an exit corner point, the entry corner point in the immediate delay successor zone is the \emph{next delay corner point}.

Let $c_1$ and $c_2$ be constraints of the form $x \succ k_1$ and $y \prec k_2$ respectively,
where $k_1, k_2 \in \{0,1, \dots, M\}$
and $\succ \:\in \: \{>, \ge \}$ and $\prec \: \in \: \{<, \le \}$.

We consider delays of the form $d_{c_1, c_2}$ and $d_{c_2, c_1}$.
$d_{c_1, c_2}$ is the delay made from an entry corner point to reach an exit corner point
while $d_{c_2, c_1}$ is the delay
made from the exit corner point to reach an entry corner point of a delay successor zone.
For example, the delay $d_{x \ge 5, y < 7}$ denotes a delay
made from a state $p'$ with $p'(x) = 5$ and after making this delay
it evolves to state $p''$ such that $p''(y) = 7 - \delta$.
Delays of the form mentioned above are called \emph{corner point delays}\index{corner point delay} or \emph{cp-delays}.

A corner point trace\index{corner point trace} (cp-trace) is a timed trace ending at a corner point through delays and visible actions such that each delay takes the trace from one corner point to the next delay corner point.
All the delay actions present in a corner point trace are cp-delays.
\end{enumerate}
\end{definition}
We create a zone graph such that for any location $l$, the zones $Z$ and $Z'$ of
any two nodes $(l, Z)$ and $(l, Z')$ in the zone graph are disjoint and all zones of the zone graph are pre-stable.
This zone graph is constructed in two phases in time exponential in the number of the clocks.
The first phase performs a forward analysis of the timed automaton while the second phase
ensures pre-stability in the zone graph.
The forward analysis may cause a zone graph to become infinite \cite{BBFL1}. Several kinds of abstractions
have been proposed in the literature \cite{DT1,BBFL1,BBLP1} to make the zone graph finite.
We use \emph{location dependent maximal constants}
abstraction \cite{BBFL1} in our construction.
In phase 2 of the zone graph creation, the zones are further split to ensure that the resultant zone graph is pre-stable.

Some approaches for preserving convexity and implementing pre-stability have been discussed in \cite{TY1}. 
As an example, consider the timed automaton in Figure \ref{fig-delayZone}. The pre-stable zones of location $l_1$ are shown in the right side of the figure.
In this paper, from now on, unless stated otherwise, \emph{zone graph} will refer to this form of the pre-stable zone graph that is described above.
An algorithmic procedure for the construction of the zone graph is given in \cite{GKNA1}.

A relation $\mathcal{R} \subseteq Q \times Q$ is a \emph{timed simulation} relation if the following conditions hold
for any two timed states $(p, q) \in \mathcal{R}$. \\
$\forall$ $a \in Act$, $p \to{a}p' \implies \exists q'$ : $q \to{a}q'$ and $(p', q') \in \mathcal{R}$ and \\
$\forall d \in \realpos$, $p \to{d}p' \implies \exists q'$ : $q \to{d}q'$ and $(p', q') \in \mathcal{R}$.\\
A \emph{timed bisimulation} relation is a symmetric timed simulation. Two timed automata are timed bisimilar if and only if their initial states are timed bisimilar. Using a product construction on region graphs,
timed bisimilarity for timed automata was shown to be decidable in EXPTIME \cite{Cer1}.
In \cite{WL1}, a product construction on zone graph was used
for deciding timed bisimulation relation for timed automata.
 \begin{figure}[t]
 \centering
 \begin{minipage}[b]{0.46\linewidth}
 \includegraphics[width=0.9\textwidth]{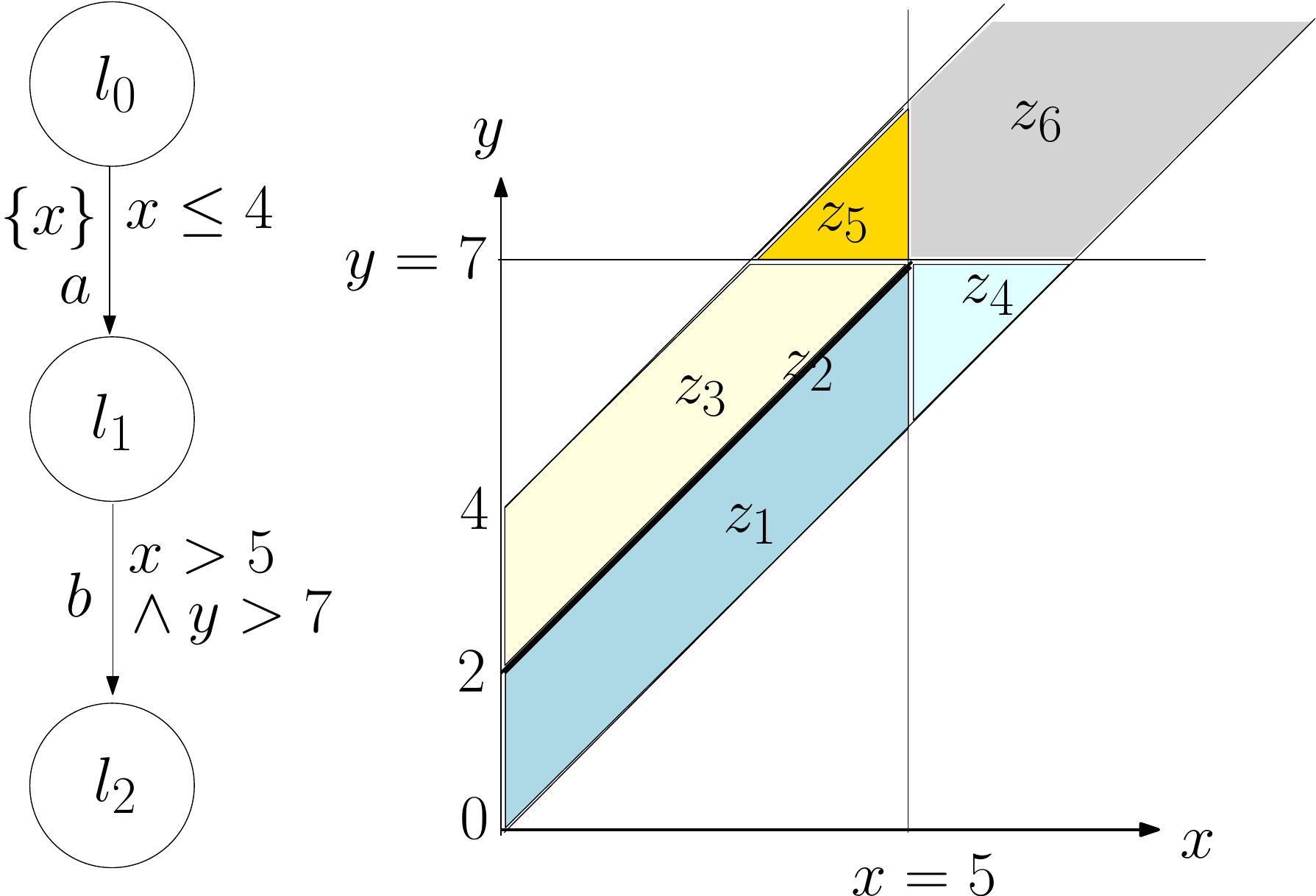}
 \caption{\label{fig-delayZone} a timed automaton and the zones for location $l_1$}
 \end{minipage}
 \begin{minipage}[b]{0.49\linewidth}
 \includegraphics[width=0.95\textwidth]{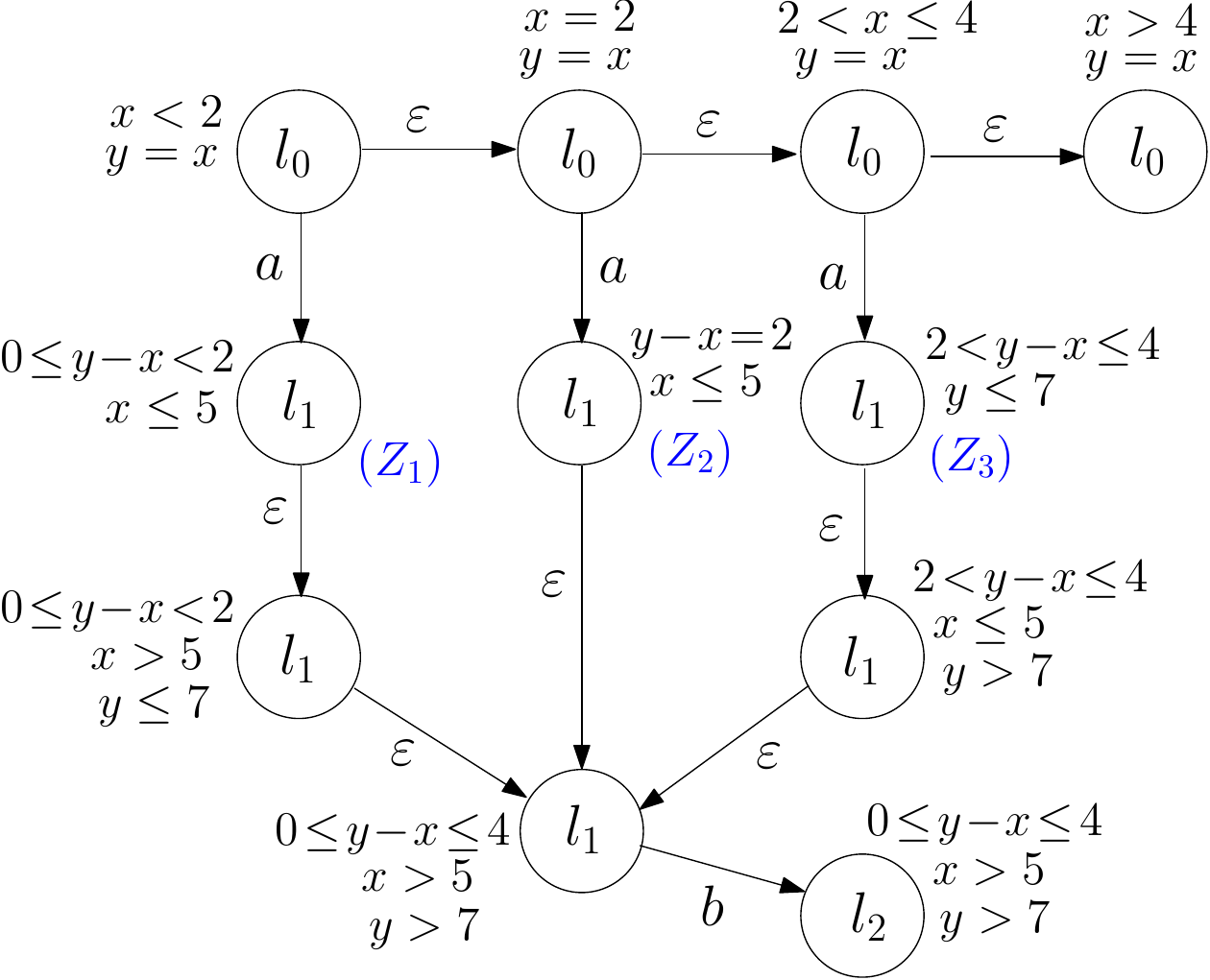}
 \caption{\label{fig-zoneGraph} zone graph for the TA in Figure \ref{fig-delayZone}}
 \end{minipage}
 \end{figure}

\section{Zone Graph}\label{sec-zoneGraph}\index{zone graph}
\begin{algorithm}
\caption{Construction of Zone Graph\newline 
\emph{Input: Timed automaton $A$ \newline
Output: Zone  graph corresponding to $A$}
}\label{algo-zonegraph}
\begin{algorithmic}[1]
{\small
\State Calculate $max_x^l$ for each location $l \in L$ and each clock $x \in C$. \Comment{\textsf{\tiny This is required 
for abstraction to ensure finite number of zones in the zone graph.}}
\State Initialize $Q$ to an empty queue.
\State $Enqueue(Q, \langle l_0, \emptyset \rangle)$. \Comment{\textsf{\tiny Every element is a pair consisting of a location and its parent}}
\State $successors\_added = false$. \Comment{\textsf{\tiny flag set to true whenever  successors of a location are added to $Q$}}
\While {Q not empty}
\State $\langle l, l_p \rangle = dequeue(Q)$
\If {$l_p\neq \emptyset$}, 
\State For the edge $l_p \to{g, a, R'} l$ in $A$, for each existing zone $Z_{l_p}$ of $l_p$, create
the zone $Z = (Z_{l_p} \uparrow \cap \:g_{[R' \leftarrow \overline{0}]})$ of $l$, when $Z \neq \emptyset$.
\State Abstract each of the newly created zones if necessary and for any newly created zone $Z$, for location $l$,
if $\exists Z_1$ of same location such that $Z \cap Z_1 \neq \emptyset$, then merge $Z$ and $Z_1$.
\State Update edges from zones of $l_p$ to zones of $l$ appropriately.
\State If {a new zone of $l$ is added or an existing zone of $l$ is modified} for all successors $l_j$ of $l$,
enqueue $\langle l_j, l \rangle$ to Q.
\State $successors\_added := true$.
\EndIf
\State $new\_zone\_l := true$. \Comment{\textsf{\tiny flag set to false when the canonical decomposition does not produce further zones}}
\While {$new\_zone\_l$}
\State Split the existing zones $Z$ of $l$ based on the canonical decomposition of the guards on the
outgoing edges of $l$ \Comment{\textsf{\tiny It is not always necessary for a split to happen.}}
\State For every zone $Z$ of $l$, consider $Z \uparrow$ and split it further based on the canonical
decomposition of the guards on the outgoing edges of $l$
\Comment{\textsf{\tiny Note that the zones created from this split are convex.}}
\State Abstract each of the newly created zones if necessary and update edges appropriately.
\If {new zones are not created} set $new\_zone\_l$ to $false$.
\EndIf
\EndWhile
\If {any new zones of $l$ are created or any existing zones of $l$ are modified due to the canonical
decomposition of the outgoing edges of $l$ and $successors\_added = false$} 
\State for all the successor locations $l_j$ of $l$ to Q, enqueue $\langle l_j, l \rangle$ to Q.
\EndIf
\EndWhile
\State  /* \textsf{Phase 2 :  In this phase, pre-stability is enforced} */
\State $new\_zone = true$
\While {$new\_zone$}
\State $new\_zone  = false$
\ForAll {edges $l_i \to{g, a, R'} l_j$}
\ForAll {pairs of zones $Z_{lik}$, $Z_{ljm}$ such that $Z_{lik} \to{\alpha}Z_{ljm}$ is an edge in the 
zone graph where $\alpha \in Act \cup \realpos$}
\If {$Z_{lik}$ is not pre-stable with respect to $Z_{ljm}$,}
Split $Z_{lik}$ to make it pre-stable with respect to $Z_{ljm}$.   \Comment{\textsf{\tiny Note that this split still maintains 
convexity of $Z_{lik}$ since the zone is split entirely along an axis that is parallel to the diagonal in 
the $\lvert C\lvert$-dimensional space.}}
\State $new\_zone := true$
\State Update the incoming and the outgoing edges for the newly created zones.
\EndIf
\EndFor
\EndFor
\EndWhile
}
\end{algorithmic}
\end{algorithm}

The detailed algorithm for creating the zone graph has been described in Algorithm \ref{algo-zonegraph}
and consists of two phases, the first one being a forward analysis of the timed automaton while the second phase
ensures pre-stability in the zone graph.
The presence of an edge to a location $l'$ in the timed automaton does not
imply that the the action of the edge can be performed or
location $l'$ is reachable.
For an edge $l \to{g, a, R} l'$, if for all zones $Z$ of $l$,
$Z \cap g = \emptyset$, then action $a$ cannot be performed at $l$
and $l'$ cannot be reached at $l$.
Such edges can be removed from the timed automaton.
For a location $l'$, if the actions corresponding to none of
the incoming edges can be performed, then the location $l'$
is unreachable and it can be removed along with its incoming and
outgoing transitions.

The set of valuations for every 
location is initially split into zones based on the canonical decomposition of its outgoing transition\delete{as discussed in
\cite{TY1}}. 

After phase 2, pre-stability ensures the following: For a node $(l, z)$ in the zone graph, with $v \in z$,
\delete{such that for a $v \models z$,}for a timed trace $tr$,
if $(l, v) \to{tr} (l'', v'')$, with $v'' \in z''$, then 
$\forall v' \; \delete{\models}\mbox{in}\; z$, $\exists tr'.(l, v') \to{tr'} (l'', \tilde{v})$, with
$untime(tr') = untime(tr)$ and $\tilde{v} \in z''$. 
\delete{Here $untime(tr)$ represents the sequence of visible actions in $tr$.}According to the construction given
in algorithm \ref{algo-zonegraph}, for a particular location of the timed automaton, the zones corresponding 
to any two nodes are disjoint. Convexity of the zones and pre-stability property \delete{of the zones}together ensure
that a zone with elapse of time is intercepted by a single hyperplane of the form $x = h$,
where $x \in C$ and $h \in \mathbb{N}$\delete{in the future}.
This is stated formally in Lemma \ref{lem-pre-stability}.
Some approaches for preserving convexity and implementing pre-stability have been discussed \delete{in detail}in \cite{TY1}. 
As an example consider the timed automaton in Figure \ref{fig-delayZone}. The zones corresponding to location $l_1$ as produced through algorithm \ref{algo-zonegraph} are shown in the right side of the figure.

For a state $q \in T(A)$, $\mathcal{N}(q)$ represents the node
 of the zone \delete{valuation}graph with the same location as that of $q$ such that the zone corresponding to $\mathcal{N}(q)$ includes 
 the valuation of $q$. We often say that a state $q = (l, v)$ is in node $s = (l, Z)$ to indicate that $v \in Z$. For two zone \delete{valuation}graphs, 
 $\mathcal{G}_{A_1} = (S_1,s_0,Lep,\rightarrow_1)$, $\mathcal{G}_{A_2} = (S_2,s_0',Lep,\rightarrow_2)$ and a relation
 $\mathcal{R}\subseteq S_1\times S_2$, $\mathcal{G}_{A_1}\:\mathcal{R}\:\mathcal{G}_{A_2}$ iff $(s_p,s_q)\in \mathcal{R}$. 
An $\varepsilon$ transition represents a delay $d \in \mathbb{R}_{\ge 0}$.
\subsection{Creating pre-stable zones}

We describe below a procedure for splitting a zone $Z_1$ into a set $\mathcal{Z}_1$ of zones w.r.t a zone $Z_2$
such that the following properties are satisfied.
\begin{enumerate}
\item Each zone $Z_i$ in $\mathcal{Z}_1$ should either be a subset of $Z_2$ or disjoint from $Z_2$.
\item The DBMs in $\mathcal{Z}_1$ are mutually exclusive and disjoint.
\item The union of the DBMs in $\mathcal{Z}_1$ is $Z_1$.
\end{enumerate}

A naive approach would be to populate $\mathcal{Z}_1$ with two DBMs, 
namely the (set) intersection and (set) difference of $Z_1$ with $Z_2$.
This ensures that the properties 1, 2 and 3 hold. 
However, while the intersection of two DBMs is always a DBM, this cannot be said for their difference. 
Thus, we seek to refine our process by replacing the difference of $Z_1$ and $Z_2$
with some DBMs which are mutually exclusive and which cover the difference of $Z_1$ and $Z_2$,
so that properties 2 and 3 continue to hold.

One naive way to achieve this is to initially populate $\mathcal{Z}_1$ with only $Z_1$,
and then iterate through all the elementary clock constraints $c_{ij}$ of $Z_2$. 
The element $m_{ij}$ in the DBM is a constraint of the form $x_i - x_j < k$ or $x_i - x_j \le k$.
In each iteration, we replace each DBM in $\mathcal{Z}_1$ with two DBMs, namely its intersection and difference with $c_{ij}$.
Note that the complement of an elementary clock constraint is always an elementary clock constraint 
and thus the difference can be calculated by simply taking the intersection 
with the complement of the elementary constraint, which implies that 
both the intersection and the difference considering a single elementary constraint as above are DBMs. 
Of course, some of the DBMs thus produced may be empty, and we do not add these to $\mathcal{Z}_1$. 
The theoretical bound on the number of DBMs in $\mathcal{Z}_1$ at the end is
$2^{\mbox{(number of elementary clock constraints in dbm2)}}$, 
that is, $2^{O((n+1)^2)}$.
The term $(n+1)^2$ comes from the fact that there are $(n+1)^2$ entries in $Z_2$ 
each of which is an elementary clock constraint. 
The above exponential bound is because of the fact that, in the worst case, 
we double the number of DBMs in $\mathcal{Z}_1$ in each iteration.
This leaves us with a set of mutually exclusive zones which satisfies properties 1, 2 and 3.

However, this is inefficient in that it generates too many zones. 
So, we refine the process of iterating through the elementary clock constraints $c_{ij}$.
We start with an empty $\mathcal{Z}_1$ and a pointer which initially points to $Z_1$. 
At each iteration, we evaluate the intersection of $c_{ij}$ and the zone pointed to by the pointer
and also their difference which too is a zone since the difference is calculated
by taking the intersection with the complement of $c_{ij}$.
The pointer is then updated to point to the intersection DBM, 
while the difference DBM is added to $\mathcal{Z}_1$. 
This procedure continues until all the $c_{ij}$ have been seen.
In the end, if the pointer's target is non-empty, we add it to $\mathcal{Z}_1$. 
At this point, the zones in $\mathcal{Z}_1$ still cover $Z_1$ and also satisfy properties 1 and 2. 
This method leads to a maximum of one zone being added to $\mathcal{Z}_1$ for each $c_{ij}$,
and one zone being added at the end, which gives us a maximum of $O((n + 1)^2 + 1)$ zones in $\mathcal{Z}_1$.

Algorithm \ref{algo-splitZone} describes the procedure mentioned above formally.

\begin{algorithm}
\caption{This algorithm returns the intersection and the difference of two zones.\newline 
\emph{Input: Two zones $Z_1$ and $Z_2$ \newline
Output: i) $Z$ which is $Z_1 \cap Z_2$ and ii) $\mathcal{Z}_1$, a set of zones whose union is $Z_1 - Z_2$}
}\label{algo-splitZone}
\begin{algorithmic}[1]
{\small
\State $Z := Z_1 \cap Z_2$
\State $\mathcal{Z}_1 := \emptyset$, $p := Z_1$ \Comment{\textsf{\tiny $p$ is a pointer pointing to a zone.}}
\ForAll {$i \in [0, n]$}
\ForAll {$j \in [0, n]$}
\State Let $c_{ij}$ be the constraint $x_i - x_j \prec m_{ij}$, where $(m_{ij}, \prec)$ is the corresponding entry in the DBM
of $Z_2$.
\State Let $\overline{c_{ij}}$ be the complement of the constraint $c_{ij}$.
\State $Z_{11} := p \cap c_{ij}$, $Z_{12} := p \cap \overline{c_{ij}}$
\If {$Z_{12} \neq \emptyset$}
\State $\mathcal{Z}_1 := \mathcal{Z}_1 \cup Z_{12}$
\EndIf
\If {$Z_{11} \neq \emptyset$}
\State $p := Z_{11}$
\EndIf
\EndFor
\EndFor
\If {$Z_{11} \neq \emptyset$} 
\State $\mathcal{Z}_1 = \mathcal{Z}_1 \cup Z_{11}$
\EndIf
}
\end{algorithmic}
\end{algorithm}

Pre-stability is ensured corresponding to both \emph{delay} and \emph{discrete} transitions in the zone graph.
\begin{itemize}
\item {\bf Pre-stabilizing zones corresponding to delay transitions:} 
Consider two nodes in the zone graph $(l, Z) \to{\varepsilon} (l, Z')$ such that $Z$ is not pre-stable w.r.t $Z'$.
The procedure involves splitting $Z$ w.r.t $Z'\!\!\downarrow$ such that each newly created zones 
obtained by splitting $Z$ is pre-stable w.r.t. $Z'$.

We follow the procedure given in Algorithm \ref{algo-splitZone} to split a zone $Z$ w.r.t the zone $Z'$.
\item {\bf Pre-stabilizing zones corresponding to discrete transitions  in the zone graph:} 
Consider two nodes in the zone graph $(l, Z) \to{a} (l', Z')$ such that $a \in Act$  and $Z$ is not pre-stable w.r.t $Z'$.
The method involving the pre-stabilization operation is similar to the $Pre_e(Z')$ operator used in backward analysis
of the timed automaton, where $e$ is an edge between say location $l$ and $l'$ of the timed automaton and $Z'$ is a zone of $l'$.
Formally for an edge $l \to{g, a, R} l'$ $Pre_e(Z')$ is defined as the following:
\begin{center}
$Pre_e(Z') = \{v \in \realpos^{\lvert C \lvert}\ \lvert \exists v' \in Z', \: \exists t \in \realpos \mbox{ such that }
v + t \models g \mbox{ and }$ \\
$[R\leftarrow 0] (v + t) = v'\}$
\end{center}
Let $\sembrack{R = 0}$ denotes the set of all clock valuations $v$
such that $v(x) = 0$ when $x \in R$ and $v(x) \in \realpos$ when $x \notin R$. $v(x)$ denotes the valuation of clock $x$.
In our case, we define the zone $\widehat{Z'}$ as follows:\\
$\widehat{Z'} = [R \leftarrow 0]^{-1}(Z' \cap \sembrack{R = 0})$.
Now $Z$ is split w.r.t $\widehat{Z'}$ following Algorithm \ref{algo-splitZone} that
makes all the newly created zones pre-stable w.r.t $Z'$.
\end{itemize}

The following lemma states an important property of the zone graph which will further be used for clock reduction.
\begin{lemma} \label{lem-pre-stability}
Pre-stability ensures that
if the zone $Z$ in any node $(l, Z)$ in the zone graph
is bounded above, then it is bounded fully from above
by a hyperplane $x = h$, where $x \in C$ and $h \in \mathbb{N}$.
\end{lemma}
\begin{proof}
We show the proof for two clocks. The same argument holds for arbitrary number of clocks.
Consider a $\lvert C \lvert$ dimensional zone $Z$.
Since $\lvert C \lvert = 2$, a zone $Z$ that is bounded above is of the form $k_{xy_2} \prec x - y \prec k_{xy_1}$,
$k_{x_1} \prec x \prec k_{x_2}$ and $k_{y_1} \prec y \prec k_{y_2}$,
where $x$ and $y$ are the two clocks, $\prec \in \{<, \le\}$.
Let $Z$ be bounded above by two hyperplanes $x = k_{x_2}$ and  $y = k_{y_2}$.
As can be seen from Figure \ref{fig-pre-proof}, there are two zones $Z_1$ and $Z_2$ such that
$Z_1$ is defined by the inequations $x-y \prec k_{xy_1}$, $k_{x_2} \prec x$ and $y \prec k_{y_2}$,
while $Z_2$ is defined as $k_{xy_2} \prec x-y$, $x \prec k_{x_2}$ and $k_{y_2} \prec y$.
Making $Z$ pre-stable will divide it into two parts:
\begin{itemize}
\item $Z'$ defined by the inequations
$k_{x_2} -  k_{y_2} \prec x - y \prec k_{xy_1}$, $k_{x_1} \prec x \prec k_{x_2}$ and $k_{y_1} \prec y$ and
\item $\widehat{Z}$ define by the inequations
$k_{xy_2} \prec x - y \prec k_{x_2} -  k_{y_2}$, $k_{x_1} \prec x$ and $y \prec k_{y_2}$
\end{itemize}
$Z'$ is bounded fully from above by the hyperplane $k_{x_2}$ while $\widehat{Z}$ is bounded fully from above by the hyperplane $k_{y_2}$.
Note that the inequations defining the zones $Z'$ and $\widehat{Z}$ may vary depending on
the relation among the various constants $k_{xy_1}, k_{xy_2}, k_{x_1}, k_{x_2}, k_{y_1}$ and $k_{y_2}$.
In all cases, pre-stability will ensure that there exists a single hyperplane that fully bounds a zone from above.
\end{proof}
\begin{figure}
\centering
\vspace{-10pt}
\includegraphics[width=0.3\textwidth]{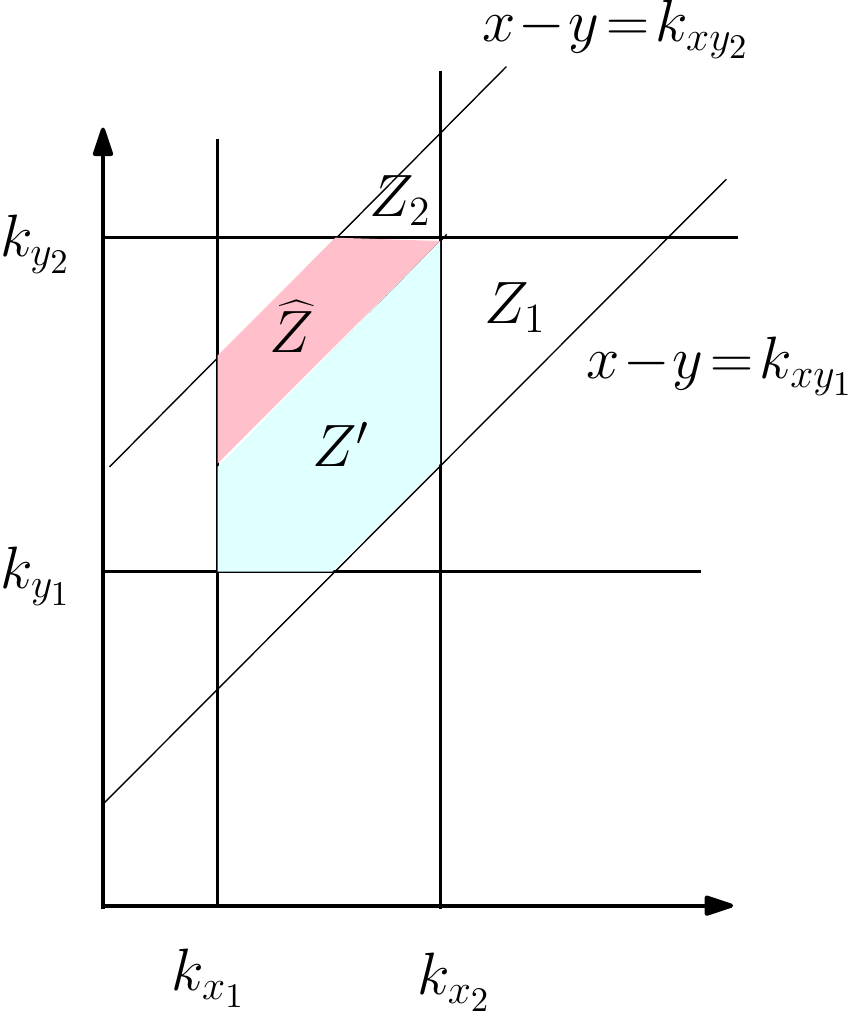}
\caption{\label{fig-pre-proof}For every zone, there exists a hyperplane that bounds it fully from above in a pre-stable zone graph}
\vspace{-10pt}
\end{figure}

\section{Clock Reduction} \label{sec-method}\index{clock reduction}
Unlike the method described in \cite{DY1}, which works on the syntactic structure of the timed automaton, we use a semantic representation as given by the zone graph described in Section \ref{sec-ta} to capture the behaviour of the timed automaton. This helps us to reduce the number of clocks in a more effective way. For a given TA $A$, we first describe a sequence of stages
to construct a TA $A_4$ that is timed bisimilar to $A$.
Later we prove the minimality in terms of the number of clocks for the TA $A_4$.
The operations involved in our procedure use a \emph{difference bound matrix} \index{difference bound matrix (DBM)} representation of the zones.
The following  are important considerations in reducing the number of clocks.
\begin{itemize}
\item There may be some clock constraints on an edge of the TA that are never enabled.
Such edges and constraints may be removed. (Stage 1)
\item Splitting some locations may lead to a reduction in the number of the clocks. (Stage 2)
\item At some location, some clocks whose values may be expressed in terms of other clock values, may be removed. (Stage 2) 
\item Two or more constraints on edges outgoing from a location
when considered collectively may lead to the removal of some constraints. (Stage 3)
\item An efficient way of renaming the clocks across all locations can reduce the total number of clocks further. (Stage 4)
\end{itemize}
Given a TA $A$, we apply the above operations in sequence to obtain the TA $A_4$. These operations are detailed below.

\textbf{Stage 1: Removing unreachable edges and associated constraints:} 
This stage involves creating the pre-stable zone graph of the given timed automaton, as described in Section \ref{sec-ta}.
As mentioned previously in that section, the edges and their associated constraints that are never enabled in an actual transition are removed while creating the zone graph.
Suppose there is an edge $l_i \to{g, a, R} l_j$ in $A$ but in the zone graph, a corresponding transition of the form $(l_i, Z) \to{a} (l_j, Z')$ does not exist.
Line 8 in Algorithm \ref{algo-zonegraph} also illustrates this.
This implies that the transition $l_i \to{g, a, R} l_j$ is never enabled and hence is removed from the timed automaton. Since the edges that do not affect any transition get removed during this stage, we have the following lemma trivially.
\begin{lemma} \label{lem-stage1}
The operations in stage 1 produce a timed automaton $A_1$ that is timed bisimilar to the original TA $A$.
\end{lemma}
\proof
During the construction of the zone graph, an edge is entirely removed if the corresponding transition is never enabled.
Removing those edges corresponding to which no transition takes place does not affect the behaviour of the TA.
Hence after stage 1, the resultant TA remains timed bisimilar to the original TA.
\qed

The time required in this stage is proportional to the size of the zone graph and hence \emph{exponential} in the number of clocks of the timed automaton.

\textbf{Stage 2: Splitting locations and removing constraints not affecting transitions:}\index{clock reduction!splitting locations}
Locations may also require to be split in order to reduce the number of clocks of a timed automaton.
Let us consider the example of the timed automaton in Figure \ref{fig-delayZone} and its zone graph in Figure \ref{fig-zoneGraph}.
There are three base zones corresponding to location $l_1$ in the zone graph, i.e.
$Z_1 = \{0 \le y - x < 2, x \le 5\}$,
$Z_2 = \{y - x = 2, x \le 5\}$ and
$Z_3 = \{2 < y -x \le 4, y \le 7\}$.
This stage splits $l_1$ into three locations $l_{1_1}$, $l_{1_2}$ and $l_{1_3}$ (one for each of the base zones $Z_1,Z_2$ and $Z_3$) as shown in Figure \ref{fig-newLoc}(a).
While the original automaton, in Figure \ref{fig-delayZone}, contains two elementary constraints on the edge between $l_1$ and $l_2$, the modified automaton, in Figure \ref{fig-newLoc}(a), contains only one of these two elementary constraints on the outgoing edges from each of $l_{1_1}$, $l_{1_2}$ and $l_{1_3}$ to $l_2$.
Subsequent stages modify it further to generate an automaton using a single clock as in Figure \ref{fig-newLoc}(b).

Splitting ensures that only those constraints, that are relevant for every valuation in the base zone of a newly created location, appear on the edges originating from that location. Since the clocks can be reused while describing the behaviour from each of the individual locations created after the split, this may lead to a reduction in the number of clocks. 
 \begin{figure}[t]
 \centering
\vspace{-10pt}
 \begin{minipage}[b]{0.48\linewidth}
 \includegraphics[width=0.95\textwidth]{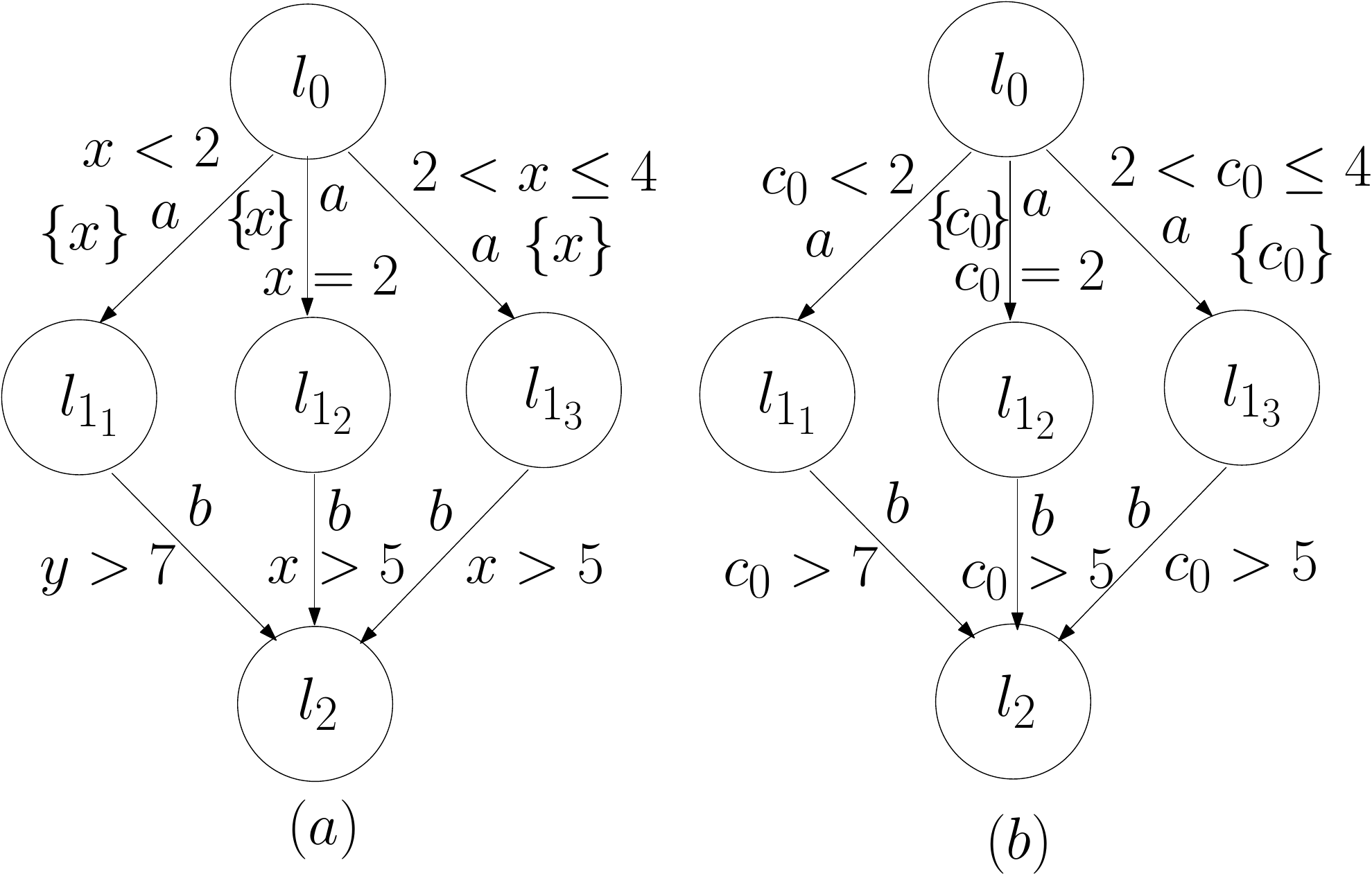}
 \caption{\label{fig-newLoc} Splitting locations of the TA in Figure \ref{fig-delayZone}}
 \end{minipage}
 \quad
 \begin{minipage}[b]{0.48\linewidth}
\includegraphics[width=1.0\textwidth]{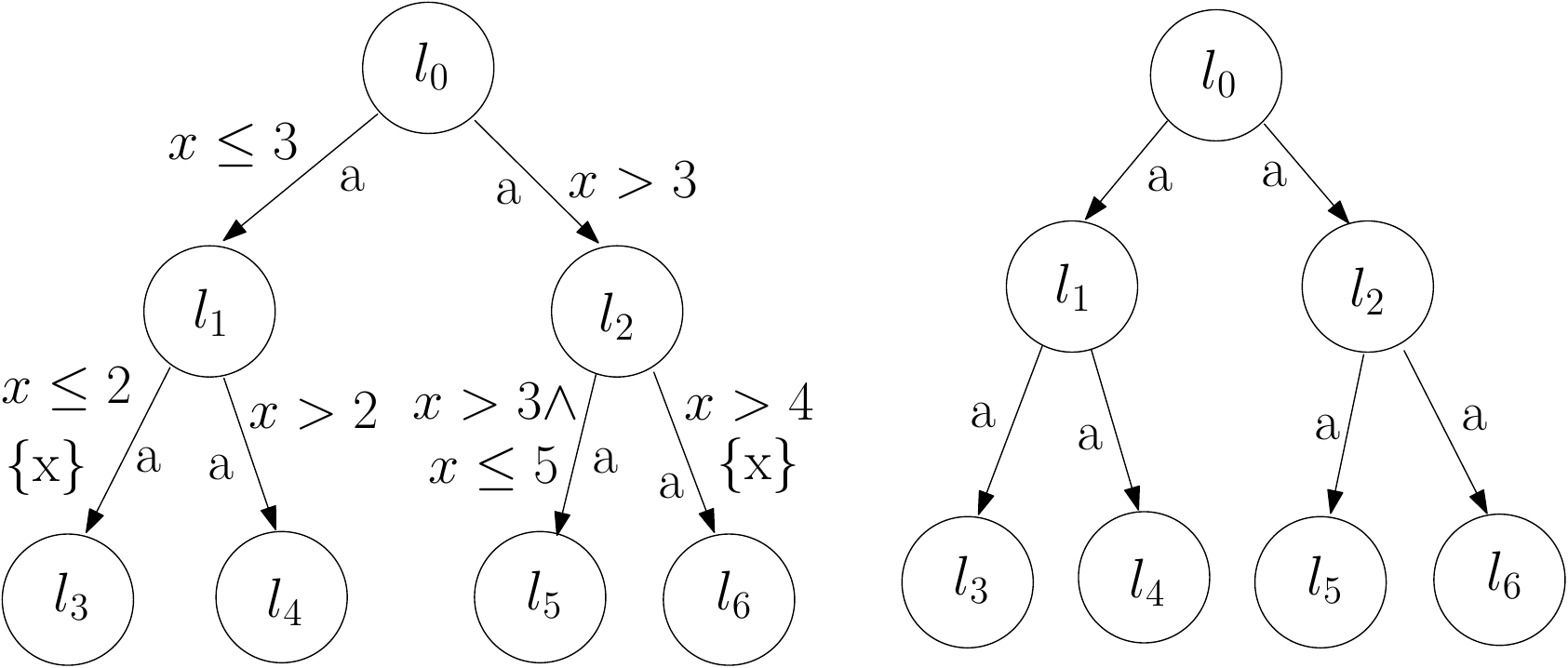}
\caption{\label{fig-stages2} The two timed automata are timed bisimilar}
 \end{minipage}
 \vspace{-10pt}
 \end{figure}
 
We describe a formal procedure for splitting a location into multiple locations in Algorithm \ref{algo-stage2}.
\begin{algorithm}[H]
\caption{Algorithm for splitting locations\newline 
\emph{Input: Timed automaton $A_1$ obtained after stage 1 \newline
Output: Modified TA $A_2$ after applying stage 2 splitting procedures}
}\label{algo-stage2}
\begin{algorithmic}[1]
{\small
\State $A_2$ := $A_1$  \Comment{\textsf{\tiny $A_1$ is the TA obtained from $A$ after the first stage}}
\ForAll {location $l_i$ in $A_1$} \Comment{\textsf{\tiny $i$ is the index of the location}}
\State Create $m$ new locations $l_{i_1}, \dots, l_{i_m}$ in $A_2$ \Comment{\textsf{\tiny  Let $m$ be the number of base zones of $l_i$}}
\State Remove location $l_i$ and all incoming and outgoing edges to and from $l_i$ from $A_2$
\ForAll {$j$ in 1 to $m$}
\ForAll{incoming edge $l_r\to{a, g_r, R_r}l_i$ in $A_1$}
\State \Comment{\textsf{\tiny Split the constraints on the incoming edges to $l_i$ for the newly created locations}}
\State $Z'_{i_j} :=  Z_{i_j} \uparrow \cap \:{\mathbb{R}_{\ge 0}^{|C|}}_{[R_r \leftarrow \overline{0}]}$ \Comment{\textsf{\tiny Let $Z_{i_j}$ be the base zone corresponding to $l_{i_j}$}}
\State \Comment{\textsf{\tiny Let  $Z_{r_j}$ is a zone of location $l_r$ from which there is an $a$ transition to $Z_{i_j}$}}
\State Let $g_{r_j}$ be the weakest formula such that 
$g_r~ \wedge~ free(Z'_{i_j}, R_r)~ \wedge ~Z_{r_j} \Rightarrow g_{r_j}$ and $g_{r_j}$ has a subset of the clocks used in $g_r$.
\Comment{\textsf{\tiny $free(Z_{i_j}',R_r)$ is the zone obtained by removing all constraints on clocks in $R_r$ in $Z_{i_j}'$}}
\State Create an edge $l_r \to{a, g_{r_j}, R_r} l_{i_j}$ in $A_2$
\EndFor

\ForAll {outgoing edge $l_i \to{a, g_{i}, R_i}l_r$ in $A_1$}
\If {$Z_{i_j} \uparrow \cap g_i = \emptyset$} 
\State Do not create this edge from $l_{i_j}$ to $l_r$ in $A_2$ since it is never going to be enabled for any valuation of $Z_{i_j}$;
\Else
\State Let $S_{r}$ be the set of elementary constraints in $g_i$ 
\Loop 
\If {$\exists s'\in S_r,~ s.t.~ Z_{i_j} \uparrow \wedge \:(\bigwedge_{s\in S_r\setminus \{s'\}} s) \Rightarrow Z_{i_j} \uparrow \land \: s'$}
\State $S_r = S_r \setminus \{s'\}$
\Else
\State Create an edge $l_{i_j} \to{a, g_{i'}, R_i} l_r$ in $A_2$, where $g_{i'}=\bigwedge_{s\in S_r}$
\State \textbf{Break};
\EndIf
\EndLoop
\EndIf
\EndFor  
\EndFor 
\EndFor 
}
\end{algorithmic}
\end{algorithm}

Note that a zone can be considered to be a set of constraints defining it.
Similarly a guard can also be considered in terms of the valuations satisfying it.
If there are $m$ base zones in $\mathcal{G}_{A_1}$ corresponding to a location $l_i$ in $A_1$,
then Line 3 and Line 4 essentially split $l_i$ into $m$ locations $l_{i_1},\cdots l_{i_m}$ in the new automaton, say $A_2$.
For each of these newly created locations, Line 6 to Line 11 determine the constraints on their incoming edges.

For each incoming edge $l_r \to{a,g_r,R_r} l_i$, there exists a zone $Z_{r_j}$ of location $l_r$
such that $Z_{r_j}$ has an $a$ transition to $Z_{i_j}$, the $j^{th}$ base zone of $l_i$.
Line 8 calculates the lower bounds of $Z_{i_j}$ by resetting the clocks $R_r$
in the intersection of $Z_{i_j} \!\!\uparrow$ with ${\mathbb{R}_{\ge 0}^{|C|}}$.
In Line 11, $free(Z_{i_j}',R_r)$ represents a zone which is obtained by removing all constraints on clocks in $R_r$ in $Z_{i_j}'$.
Further, $g_{r_j}$ is calculated as the weakest guard that is simultaneously satisfied
by the constraints $g_r$, $Z_{r_j}$ and $free(Z_{i_j}',R_r)$  and has a subset of the clocks appearing in $g_r$. 
For our running example in Figure \ref{fig-delayZone}, if we consider $Z_{i_j}=Z_1$ then 
we have $Z_{i_j}' = Z_{i_j}\!\!\uparrow \cap~ \mathbb{R}^{\{x, y\}}_{[x \leftarrow \bar{0}]} = \{x = 0, y < 2\}$, $Z_{r_j}=\{x=y, x<2\}$ and $free(Z_{i_j}', \{x\}) = \{x \ge 0, y < 2\}$. We can see that $x< 2$ is the weakest formula such that $x\le 4 \land x \ge 0 \land y<2 \land x=y \land x<2 \Rightarrow x < 2$ holds and hence $g_{r_j}=\{x < 2\}$.

The loop from Line 14 to Line 28 determines the constraints on the outgoing edges from these new locations.
Line 15 checks if the zone $Z_{i_j}\!\!\uparrow$ has any valuation that
satisfies the guard $g_i$ on an outgoing edge from location $l_i$.
If no satisfying valuation exists then this transition will never be enabled from $l_{i_j}$ and hence this edge is not added in $A_2$. Loop from Line 19 to Line 26 checks if some elementary constraints of the guard are implied by other elementary constraints of the same guard. If it happens then we can remove those elementary constraints from the guard
that are implied by the other elementary constraints.

For our running example, the modified automaton of Figure \ref{fig-newLoc}(a) does not contain the constraint $x>5$ on the edge from $l_{l_1}$ to $l_2$ even though it was present on the edge from $l_1$ to $l_2$. The reason being that the future of the zone of $l_{l_1}$ (that is $0 \le y-x<2$) along with the constraint $y>7$ implies $x>5$ hence we do not need to put $x>5$ explicitly on the outgoing edge from $l_{1_1}$ to $l_2$.
Removal of such elementary constraints helps in reducing the number of clocks in future stages.
The maximum number of locations produced in the timed automaton as a result of the split is bounded by the number of zones in the zone graph. This is exponential in the number of clocks of the original TA $A$.
However, we note that the base zones of a location $l$ in the original TA are distributed
across multiple locations as a result of the split of $l$
and no new valuations are created. This gives us the following lemma.
\begin{lemma} \label{lem-split}
The splitting procedure described in this stage does not increase the number of clocks
in $A_2$, but the number of locations in $A_2$ may become exponential in the
number of clocks of the given TA $A$.
However, there is no addition of new valuations
to the underlying state space of the original TA $A$ and corresponding to every state
$(l	, v)$ of a location $l$ in the TLTS of the original TA $A$,
exactly one state $(l_i, v)$ is created in the TLTS of the modified TA $A_2$,
where $l_i$ is one of the newly created locations as a result ofhttp://www.zalafilms.com splitting $l$.
\end{lemma}
Splitting locations and removing constraints as described above do not alter the behaviour of the timed automaton that leads us to the following lemma.
\begin{lemma} \label{lem-stage2}
The operations in stage 2 produce a timed automaton $A_2$ that is timed bisimilar to the TA $A_1$ obtained at the end of stage 1.
\end{lemma}
\proof
In the second stage, a location $l$ is split such that corresponding to
every base zone of $l$, new locations
$l_1, \dots, l_n$ are created.
Consider valuations $v_1, \dots, v_n$  in $l_1, \dots, l_n$ respectively.
Since the zones of $l$ are distributed over locations $l_1, \dots, l_n$,
before the split, all the valuations $v_1, \dots, v_n$ existed in $l$.
For all $v_i$, $1 \le i \le n$, $(l, v_i)$ (before $l$ is split) is
timed bisimilar to $(l_i, v_i)$ after the split.
This is because removal of constraints that are implied by other constraints and
splitting a location does not change the behaviour of a state $(l, v_i)$.
Hence the timed automaton after stage 2 remains timed bisimilar to the TA $A_1$.
\qed

The number of locations after the split increases exponentially
in the worst case in the number of the clocks.
The constraints on the incoming edges of $l$ are also split appropriately
into constraints on the incoming edges of the newly created locations.
Thus these operations require exploring the zones corresponding to location $l_i$
and the incoming and the outgoing edges of $l_i$ and their constraints on these edges.
Hence this stage too runs in time that is \emph{exponential} in the number of the clocks of the timed automaton.

\textbf{Stage 3: Removing constraints by considering multiple edges with the same action:}\index{clock reduction!merging guards}
We consider the example in Figure \ref{fig-stages2}. Note that the constraints $x \le 3$ and $x > 3$ on the edges from $l_0$ to $l_1$ and from $l_0$ to $l_2$ respectively could as well be merged together to produce a constraint without any clock.

For every action $a$ enabled at any location $l$, this stage checks whether
a guard enabling that action at $l$ can be merged with another guard enabling 
the same action at that location such that timed bisimilarity is preserved.
The transformation made in this stage has been formally described in Algorithm \ref{algo-stage3}.
The input to this algorithm is the TA obtained after stage 2, say $A_2$.
For each location $l_i$, the algorithm does the following: for every action $a \in Act$, it determines the zones of $l_i$
from which action $a$ is enabled. We call this set ${\mathcal{Z}_i}_{a}$. Zone graph construction and
splitting of locations in stage 2 ensures that 
all zones in ${\mathcal{Z}_i}_a$ form a linear chain connected by $\varepsilon$ edges as shown in Figure \ref{fig-zoneConstraints}. We use $\before$ to capture this total ordering relation.
Let $Z_{i_1}\before \cdots \before Z_{i_{k}} \before Z_{i_{k+1}} \cdots \before Z_{i_m}$
be the zones in this chain.
Lemma \ref{lem-pre-stability} ensures that for each $Z_{i_k}$, $k\ge 1$, that is bounded above,
there exists a hyperplane that bounds the zone fully from above
and similarly, for each $Z_{i_k}$, $k>1$
there exists a hyperplane that bounds the zone fully from below.
For a zone $Z$, let $\lowerbound{Z}$ and $\upperbound{Z}$ denote these lower and upper bounding hyperplanes of $Z$ respectively.
Further, $\upperbound{Z}$ is $\infty$ if $Z$ is not bounded from above.

Let $\setg{l_i}{a}=\{g\mid l_i\to{a,g,R} l'\in E_{A_2}\}$ be the set of guards on the outgoing edges from $l_i$ in $A_2$ which are labelled with $a$. For any $g\in \setg{l_i}{a}$, we define the following:
\begin{itemize}
\item $\starts{g}{(l_i,a)}= Z \in {\mathcal{Z}_i}_a$ is the zone in ${\mathcal{Z}_i}_a$ which is bounded from below by the same constraints as the lower bound of the constraints in $g$. 
\item $\send{g}{(l_i,a)}=Z \in {\mathcal{Z}_i}_a$ is the zone in ${\mathcal{Z}_i}_a$ which is bounded from above by the same constraints as the upper bound of the constraints in $g$.
If the constraints in $g$ do not have any upper bound,
then $\send{g}{(l_i,a)}$ is the last zone in the chain ${\mathcal{Z}_i}_a$.
\item $\range{g}{(l_i,a)}=\{Z \in {\mathcal{Z}_i}_a \mid \starts{g}{(l_i,a)} \before Z \land Z \before \send{g}{(l_i,a)}\}
\cup \{\starts{g}{(l_i,a)}\} \cup \{\send{g}{(l_i,a)}\}$ 
is the set of zones ordered by $\before$ relation in between $\starts{g}{(l_i,a)}$ and $\send{g}{(l_i,a)}$.
\end{itemize}
In Algorithm \ref{algo-stage3}, we use a rather informal notation $g:=[C1,C2]$ to denote that $C1$ and $C2$ are the constraints defining the lower and the upper bounds of $g$ respectively.
If $g$ does not have any constraint defining the upper bound then $C_2 = \infty$.
We define a total order $\beforeg$ on $\setg{l_i}{a}$ such that 
for any $g,g' \in \setg{l_i}{a}$, $g \beforeg g'$ iff $\exists Z \in \range{g}{(l_i,a)}$ such that $Z \before Z'$ for all $Z' \in \range{g'}{(l_i,a)}$.
Similar to the zones let us use ordered indexed variable $g_{i_1},\dots, g_{i_p}$ to denote $g_{i_1}\beforeg \cdots g_{i_{k}} \beforeg g_{i_{k+1}} \cdots \beforeg g_{i_{p}}$. One such total order on guards is shown in Figure \ref{fig-zoneConstraints}.
The loop from Line 5 to Line 33 in Algorithm \ref{algo-stage3} traverses the elements of $\setg{l_i}{a}$ in this total order with the help of a variable $next$ initialized to 2. 
In every iteration of this loop, the invariant $g_{curr} \beforeg g_{i_{next}}$ holds. Three possibilities exist based on whether the set union of zones corresponding to these guards is (i) not convex (ii) convex but non-overlapping, or (iii) convex as well as overlapping.

If the union is non-convex then both $g_{curr}$ and the index $next$ are changed in Line 7 to pick the next ordered pair in this order.
For cases (ii) and (iii), new guards are created by merging corresponding zones as long as the modified automaton preserves timed bisimilarity.
If timed bisimilarity is preserved then the modified automaton $A'$ is set as the current automaton which is $A_3$ (Line 13 and Line 28)
and $next$ is incremented to process the next guard.
Otherwise the guard $g_{curr}$ is set to $g_{i_{next}}$ and $next$ is incremented by 1 (Line 15 and Line 33).
The only difference in these two cases is in creating the new guard. 

For case (ii), convex but non-overlapping zones, a new guard is created from the lower bound of $\starts{g_{curr}}{(l_i,a)}$ and the upper bound of $\send{g_{i_{next}}}{(l_i,a)}$.
For case (iii), there are three possibilities of combining guards, mentioned in Line 19, Line 21 and Line 23.
The first possibility is the same as in case (ii). The second and the third possibilities are replacing the upper bound of $g_{curr}$ with the lower bound of $\starts{g_{i_{next}}}{(l_i,a)}$ and the lower bound of $g_{i_{next}}$ with the upper bound of $\send{g_{curr}}{(l_i,a)}$ respectively.

A zone graph captures the behaviour of the timed automaton and hence timed bisimilarity between two TAs can be checked using their zone graphs \cite{WL1,GKNA1}. This is why we create the pre-stable zone graph as described in Section \ref{sec-ta} as it enables one to directly check timed bisimilarity on this zone graph created through Algorithm \ref{algo-zonegraph} following the procedure described in \cite{GKNA1}.
 
\begin{algorithm}
\caption{Algorithm for stage 3\newline 
\emph{Input: Timed automaton $A_2$ obtained after stage 2 \newline
Output: Modified TA $A_3$ after applying stage 3 procedures}
}\label{algo-stage3}
\begin{algorithmic}[1]
{\small
\State $A_3$ := $A_2$  \Comment{\textsf{\tiny $A_2$ is the TA obtained from $A$ after the first two stages}}
\ForAll {location $l_i$ in $A_2$} \Comment{\textsf{\tiny $i$ is the index of the location, the set of locations do not change in this stage}}
\ForAll {$a \in sort(l_i)$}  \Comment{\textsf{\tiny $sort(l_i)$ is the set of actions in $l_i$ that can be performed from $l_i$}}
\State $g_{curr}:=g_{i_1}$, $next:=2$
\While {$next < \lvert {\setg{l_i}{a}} \lvert-1$}
\If {$\range{g_{curr}}{(l_i,a)} \cup \range{g_{i_{next}}}{(l_i,a)}$ is not convex} 
  \State $g_{curr}:=g_{i_{next}}$, $next:=next+1$
\ElsIf {$\range{g_{i_{next}}}{(l_i,a)} \cap \range{g_{curr}}{(l_i,a)} = \emptyset$} \Comment{\textsf{\tiny non-overlapping but contiguous}}  
  \State $g_{curr}':=[\lowerbound{\starts{g_{curr}}{(l_i,a)}},\upperbound{\send{g_{i_{next}}}{(l_i,a)}}]$
  \State $g_{i_{next}}':=g_{curr}'$ 
  \State Let $A'$ be the TA obtained by replacing all occurrences of $g_{curr}$ and $g_{i_{next}}$ with $g_{curr}'$ and $g_{i_{next}}'$ respectively in $A_3$
  \If {$A'$ is timed bisimilar to $A_3$}
    \State $A_3:=A'$, $g_{curr}:=g_{i_{next}}'$, $next:=next+1$
  \Else 
    \State $g_{curr}:=g_{i_{next}}$, $next:=next+1$
  \EndIf
\Else \Comment{\textsf{\tiny{$\range{g_{curr}}{(l_i,a)}$ and $\range{g_{i_{next}}}{(l_i,a)}$ have overlapping zones}}}
  \State \Comment{\textsf{\tiny{There are three ways to combine $g_{curr}$ and $g_{i_{next}}$, and}}}
  \State \Comment{\textsf{\tiny{Resultant new guards should be checked for timed bisimilarity in the following order}}}
  \State (i). $g_{curr}':=[\lowerbound{\starts{g_{curr}}{(l_i,a)}}, \upperbound{\send{g_{i_{next}}}{(l_i,a)}}]$, \\
  $\qquad \qquad \qquad \qquad g_{i_{next}}':=g_{curr}'$ 
  \State (ii). $g_{curr}':=[\lowerbound{\starts{g_{curr}}{(l_i,a)}}, \lowerbound{\starts{g_{i_{next}}}{(l_i,a)}}]$, \\
  $\qquad \qquad \qquad \qquad g_{i_{next}}':=[\lowerbound{\starts{g_{i_{next}}}{(l_i,a)}},\upperbound{\send{g_{i_{next}}}{(l_i,a)}}]$ 
  \State (iii). $g_{curr}':=[\lowerbound{\starts{g_{curr}}{(l_i,a)}},\upperbound{\send{g_{curr}}{(l_i,a)}}]$, \\
  $\qquad \qquad \qquad \qquad g_{i_{next}}':=[\upperbound{\send{g_{curr}}{(l_i,a)}}, \upperbound{\send{g_{i_{next}}}{(l_i,a)}}]$ 
  \While {$1\le i \le 3$}  \Comment{\textsf{\tiny Corresponding to the three cases above}}
    \State Let $A'$ be the TA obtained by replacing all occurrences of $g_{curr}$ and $g_{i_{next}}$ with the $i^{th}$ $g_{curr}'$ and $g_{i_{next}}'$ respectively in $A_3$
    \If {$A'$ is timed bisimilar to $A_3$}
      \State $A_3:=A'$, $g_{curr}:=g_{i_{next}}'$, $next:=next+1$
      \State \textbf{Break}
    \Else 
      \State $i:=i+1$
    \EndIf
  \EndWhile
  \If {$i=4$} \Comment{\textsf{\tiny{Bisimilarity could not be preserved in any of these three cases}}}
       \State $g_{curr}:=g_{i_{next}}$, $next:=next+1$
  \EndIf
\EndIf
\EndWhile
\EndFor  
\EndFor 
}
\end{algorithmic}
\end{algorithm}

\begin{figure}[t]
\centering
\begin{minipage}[b]{0.55\linewidth}
\includegraphics[width=0.98\textwidth]{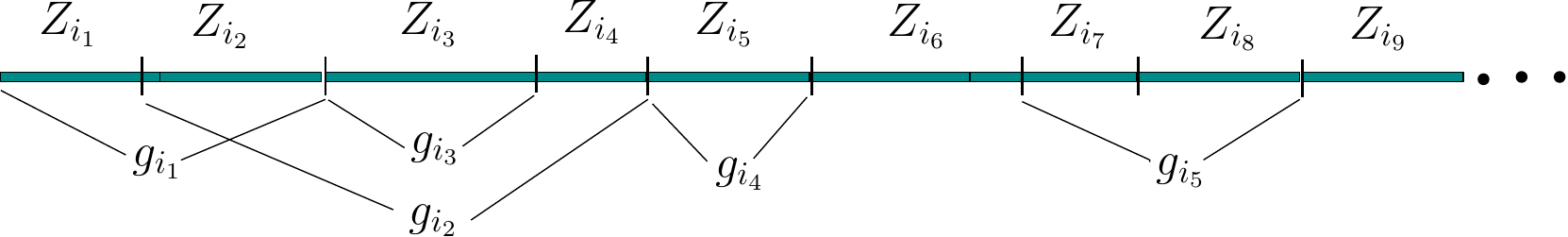}
\caption{\label{fig-zoneConstraints} Removing constraints in stage 3}
\end{minipage}
\quad
\begin{minipage}[b]{0.4\linewidth}
\includegraphics[width=0.95\textwidth]{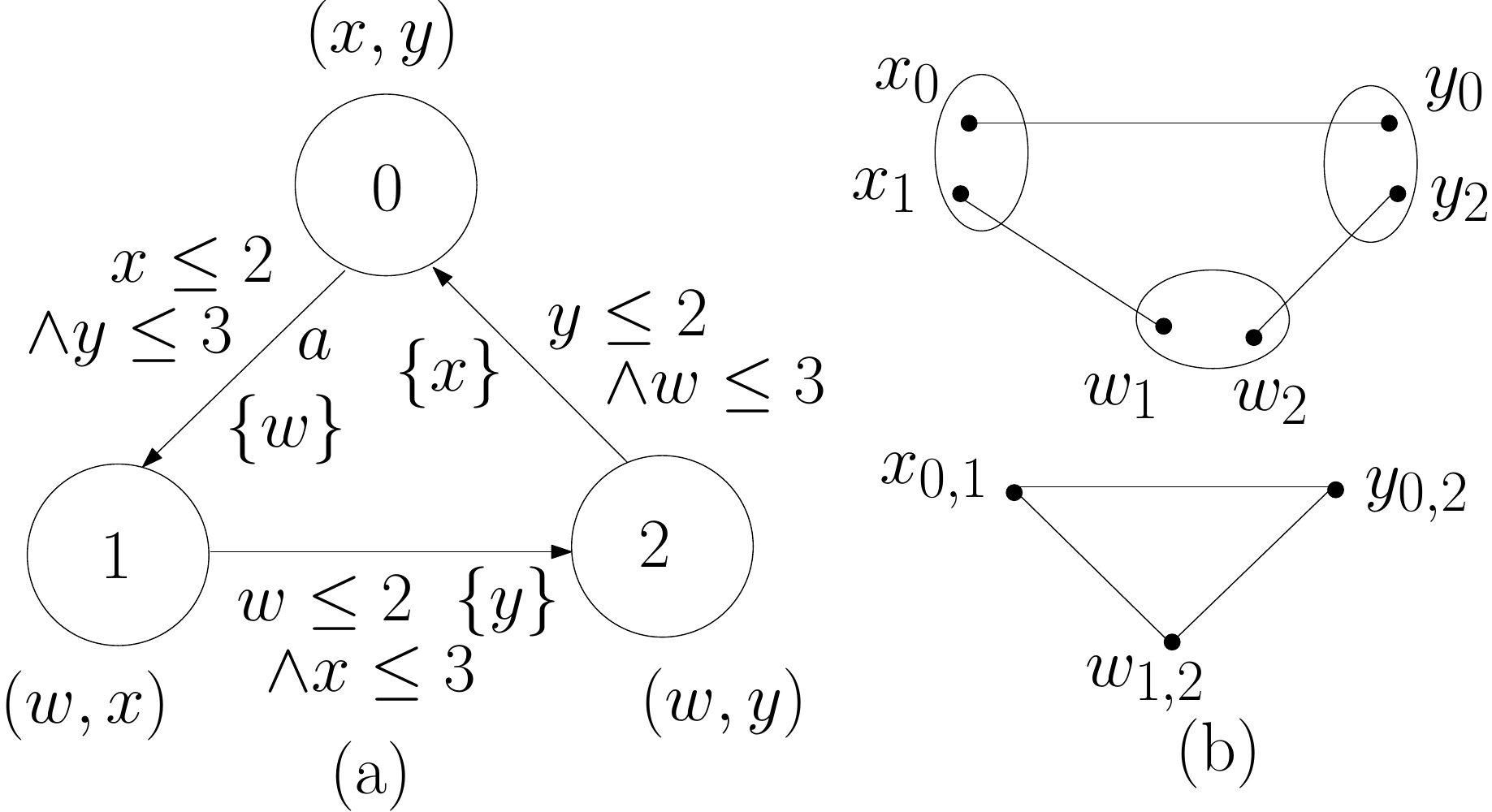}
\caption{\label{fig-clockcol} Colouring clock graph}
\end{minipage}
\end{figure}
\begin{lemma} \label{lem-stage3}
The operations in stage 3 produce a timed automaton $A_3$ that is timed bisimilar to the TA $A_2$ obtained at the end of stage 2.
\end{lemma}
\proof
In the third stage, a zone of some location $l$ may be merged with another zone if their union is convex
and this process may be repeated a finite number of times.
The successor zones are also re-computed following the constraints on the outgoing edges from location $l$.
\delete{This leads to the removal of the hyperplanes bounding the zones.}This operation is done if the resultant zone graph is still timed bisimilar to the original zone graph, i.e. the initial states of the two zone graphs are timed bisimilar. The changes in the zone graph are reflected in the timed automaton in the following way. If from a location $l_i$  there are edges of the form $l_i \to{a, g_1, R_1} l_{i1}$, $l_i \to{a, g_2,R_2} l_{i2}$, $\dots, l_i \to{a, g_n, R_n} l_{in}$ in the timed automaton $A$ such that they are labelled with the same action $a$, then some of the constraints may be replaced with a constraint equivalent to $g_i \cup g_j \cup \dots$, $g_i, g_j \in \{g_1, g_2, \dots, g_n\}$
if the resultant timed automaton is timed bisimilar to the original timed automaton. Hence the operations in this stage create a timed bisimilar TA.
\qed

The procedure above requires checking whether the union of two zones is convex or not.
Here we describe the procedure as outlined in \cite{BFT1} for this.
An \emph{H-polyhedron}\index{H-polyhedron} defines the half-space representation of a convex polyhedron.
Let $\mathfrak{P}$ be a convex polyhedron in $\mathbb{R}^{dim}$, $dim$ being the dimension in the real space.
Given $\lambda \in \mathbb{R}^{dim}$ and $\chi \in \mathbb{R}$, the inequality $\lambda^Tx \le \chi$
is said to be \emph{valid} for $\mathfrak{P}$ if the inequality is satisfied by all points in $\mathfrak{P}$.

Let
$\mathfrak{P} = \{x \in \mathbb{R}^{dim} : \mathfrak{A}x \le \mu\}$ and \\
$\mathfrak{Q} = \{x \in \mathbb{R}^{dim} : \mathfrak{B}x \le \nu\}$
be (possibly unbounded) H-polyhedra\index{H-polyhedra}
whose envelope is defined by
\begin{center}
$env(\mathfrak{P}, \mathfrak{Q}) = 	\{x \in \mathbb{R}^{dim} : \overline{\mathfrak{A}}x \le \overline{\mu}, \overline{\mathfrak{B}}x \le \overline{\nu}\}$,
\end{center}
where $\overline{\mathfrak{A}}x \le \overline{\mu}$ is the subsystem of $\mathfrak{A}x \le \mu$
obtained by removing all the inequalities not valid for the polyhedron $\mathfrak{Q}$.
Similarly, $\overline{\mathfrak{B}}x \le \overline{\nu}$ is the subsystem of $\mathfrak{B}x \le \nu$
obtained by removing all the inequalities not valid for the polyhedron $\mathfrak{P}$.
It is easy to see that $\mathfrak{P} \cup \mathfrak{Q} \subseteq env(\mathfrak{P}, \mathfrak{Q})$.
Theorem 3 in \cite{BFT1} states the following important result.
\begin{center}
$\mathfrak{P} \cup \mathfrak{Q}$ is convex $\Leftrightarrow \mathfrak{P} \cup \mathfrak{Q}  = env(\mathfrak{P}, \mathfrak{Q})$
\end{center}
This theorem also leads to an algorithm for checking the convexity of $\mathfrak{P} \cup \mathfrak{Q}$ for two given polyhedra $\mathfrak{P}$ and $\mathfrak{Q}$
and it also produces a minimal H-representation of $\mathfrak{P} \cup \mathfrak{Q}$ when the union is convex.
The representation is minimal in the sense that it does not contain any redundant inequality.
An inequality is said to be redundant for an $H\!\!-\!\!polyhedron$ $\mathfrak{P}$,
if its removal preserves the polyhedron.
As shown in the paper \cite{BFT1}, the complexity of running the algorithm is polynomial
though it is not strongly polynomial since the algorithm involves solving LPs	
but algorithms for solving LPs like the ellipsoid method and the interior point method are not strongly polynomial.

As mentioned above, in this stage, while removing the constraints, timed bisimilarity is checked and the number of bisimulation checks is bounded by the number of zones in the zone graph of the TA obtained after stage 2.
Further each check of timed bisimilarity involves checking convexity of union of zones
that invokes the algorithm of \cite{BFT1} that has a polynomial complexity.
From Lemma \ref{lem-split}, in stage 2, no new valuations are added to the underlying state space of the original TA $A$,
and corresponding to every valuation $(l, v)$, exactly one valuation $(l_i, v)$ is created,
where $l$ and $l_i$ are as described in Lemma \ref{lem-split}. Thus the number of zones in the zone graph of $A_2$
is still exponential in the number of clocks of the original TA $A$.
Checking timed bisimilarity is done in EXPTIME \cite{Cer1,LH1}. A zone graph is constructed prior to every bisimulation check and
the construction is done in EXPTIME. Hence this entire stage runs in EXPTIME.

\textbf{Stage 4: Finding Active clocks, clock replacement and renaming:}
The active clocks in a location determine the behaviour
of the timed automaton from that location.
The valuation of a clock that is not active at location $l$
does not affect the behaviour of the timed automaton from that location.
For example, a clock that does not appear in the
outgoing edges of the locations reachable from $l$
is not active at $l$.
Similarly, if a clock is reset on every path starting from location $l$ 
before appearing in a constraint, then too it is not active at $l$.
Given a location $l$, an iterative method for finding the set of active clocks at $l$, denoted $act(l)$, is given in \cite{DY1}.
The method has been modified and stated below for the case where clock assignments of the form $x := y, x, y \in C$ are disallowed.

\textit{\underline{Determining active clocks}}\index{active clocks} : For a location $l$, let $clk(l)$ 
be the set of clocks that appear on the constraints on the outgoing edges of $l$.
Let $\rho : (2^C \times E) \rightarrow 2^C$ be a partial function such that for an edge $e = l \to{g, a, R} l'$, $\rho(act(l'), e)$ gives the set of active clocks of $l'$ that are not reset along $e$.
For all $l \in L$, $act(l)$ is the limit of the convergent sequence $act_0(l) \subseteq act_1(l) \dots$ such that
$act_0(l) := clk(l)$ and 
$act_{i+1}(l) := act_i(l) \: \cup \: \displaystyle{\bigcup_{e=(l, g, a, R, l') \in E}} \rho(act_i(l'), e)$.

\textit{\underline{Removing redundant resets}} : Once we find the active clocks of a location $l$, we remove all resets of clock $x$ on the incoming edges of $l$ if $x \notin act(l)$.

\textit{\underline{Partitioning active clocks}}\index{active clocks!partitioning} : Using the DBM representation of the zones, one can determine from the set of active clocks in every location whether some of the clocks in the timed automaton can be expressed in terms of other clocks and thus be removed.
Any $x,y\in act(l)$ belong to an equivalence class iff the same relation of the form $x-y=k$,
for some fixed integer $k$, is maintained between these clocks across all zones of $l$.
This is checked using the DBM of the zones of $l$. In this case either $x$ can be replaced by $y + k$ or
$y$ can be replaced by $x - k$. Let $\pi_l$ be the partition induced by this equivalence relation. 

We note that the size of the largest partition does not give the
minimum number of clocks required to represent a TA while preserving timed bisimulation.
An example is shown in Figure \ref{fig-clockcol}(a).
Though the automaton in the figure has two active clocks
partitioned into two different classes in every location,
a timed bisimilar TA cannot be constructed with only two clocks.
Assigning the minimum number of clocks to represent the timed automaton so that timed bisimilarity is preserved can be reduced to the problem of finding the chromatic number of a graph as described below.

\textit{\underline{Clock graph colouring and clock renaming}}\index{clock reduction!clock renaming} : A clock graph\index{clock graph}, $G_{A_3}$, for the timed automaton
$A_3$ is constructed in the following way. This graph contains a vertex for each class in the partition $\pi_l$, for every location $l$. Let $V_l$ be the set of vertices corresponding to the classes of $\pi_l$.
For each pair of distinct vertices $r_1, r_2 \in V_l$, an edge between $r_1$ and $r_2$ exists denoting that $r_1$ and $r_2$ cannot be assigned the same colour.
This is because two classes in the partition $\pi_l$ cannot be represented using the same clock.

Moreover, if at least one clock, say $x$, is common in two classes corresponding to two different locations without any intervening reset of $x$ then only one vertex represents these two classes.
For example, in Figure \ref{fig-clockcol}, clock $x$ is active in both locations $0$ and $1$.
$\{x\}$ forms a class in the partition of the active clocks for each of locations $0$ and $1$.
Thus we create vertices $x_0$ and $x_1$ corresponding to these two classes.
However, since there is no intervening reset of clock $x$ between locations 0 and 1,
the vertices $x_0$ and $x_1$ are merged together
and the resultant graph is termed the \emph{clock graph}.
Thus after merging some classes into one class,
the resultant class can have active clocks corresponding to multiple locations.
For a class $\mathcal{T}$, let $loc(\mathcal{T})$ represent the set of locations
whose active clocks are members of $\mathcal{T}$.

Finding the minimum number of clocks to represent the TA $A_3$ is thus equivalent to colouring its clock graph\index{clock graph!colouring} with the minimum number of colours so that no two adjacent vertices have the same colour. The number of colours gives the minimum number of clocks required to
represent the TA. If a colour $\iota$ is assigned to a vertex $r$, then all the clocks in the 
class corresponding to $r$, say $\mathcal{T}$, are renamed $\iota$.
The value of $\iota$ can be chosen to be equal to some
clock in $\mathcal{T}$ that is considered to be the \emph{representative clock} for that class.
The constraints involving the rest of these clocks in $\mathcal{T}$
are adjusted appropriately and any resets of the clocks, different from the representative clock, present on the incoming edges to $l$ such that $l \in loc(\mathcal{T})$ are also removed.

For example, suppose vertex $r$ corresponds to a class $\mathcal{T}$ having clocks $x$, $y$ and $z$
such that the valuations of the clocks are related as : $x - y  = k_1$ and $y - z = k_2$.
If colour $\iota$ is assigned to vertex $r$, then the clocks $x$, $y$ and $z$ in class $\mathcal{T}$ are replaced with $\iota$. If the value of clock $\iota$ is chosen to be the same as clock $y$, then every occurrence of $x$ in $\mathcal{T}$ is replaced with $y + k_1$,
while every occurrence of $z$ in $\mathcal{T}$ is replaced with $y - k_2$
in  the constraints involving $x$ and $z$.
The corresponding resets of clocks $x$ and $z$ are also removed.

In Figure \ref{fig-clockcol}(a), a TA with three locations is shown. In locations 0, 1 and 2, the sets of active clocks are $\{x, y\}$, $\{w, x\}$ and $\{w, y\}$ respectively. At every location, in this example, each of the active clocks itself makes a class of the partition.
Since there are six classes in total, we draw initially six vertices.
As mentioned earlier, the vertices $x_0$ and $x_1$ are merged into a single vertex.
Similarly $w_1$, $w_2$ and $y_0$, $y_2$ are also merged. We call the resultant vertices $x_{0,1}$, $w_{1,2}$ and $y_{0,2}$.
Adding the edges as described previously, we get the clock graph which is a triangle as shown in Figure \ref{fig-clockcol}(b).
Thus the chromatic number of this graph is 3 which translates to the
number of clocks obtained through the operations in stage 4.
Since this stage consists of finding the active clocks and renaming them, we have the following lemma.
\begin{lemma} \label{lem-stage4}
The operations in stage 4 produce a timed automaton $A_4$ that is timed bisimilar to the TA $A_3$ obtained after stage 3.
\end{lemma}
\proof
In stage 4, clocks are renamed.
Also in every zone in each location, the active clocks are identified and partitioned such that all clocks belonging to a class in the partition are represented with a single clock, say $\iota$.
If there is a clock $x$ in the same class $\mathcal{T}$ such that $x - \iota = m$, where $m$ is an integer, then the constraint of the form $x \bowtie k$ is replaced with $\iota \bowtie k -m$ for all the constraints on the edges outgoing from locations $l \in loc(\mathcal{T})$.
The clock renaming and replacing the integer constants in the way mentioned above do not change any bisimulation property of a timed state.
Hence the operations in stage 4 preserve the timed bisimulation property of the original TA $A$.
\qed
\begin{lemma} \label{lem-clockrename}
The problem of clock renaming operation on the clock graph as described above is NP-complete in the size of the input.
\end{lemma}
\proof Consider any arbitrary graph $G$ without any parallel edge, i.e. there should be at most one edge between a pair of vertices in $G$.
We show that there exists a TA whose clock graph is the same as $G$. Thus coloring the vertices of $G$ is reduced to  assigning minimum number of clocks to the clock graph.

The reduction is the following: Let the vertices of $G$ be $x_1, x_2, \dots, x_n$.
Corresponding to every edge $e$ between $x_i$ and $x_j$, $1 \le i, j \le n$,
create a location $l$ of the timed automaton whose active clocks are $x_i$ and $x_j$.
For the reduction to be correct, we also need to show that $x_i$
cannot be replaced with $x_j$ in any of the constraints on the outgoing edges from location $l$.
Let the locations of the timed automaton be numbered $l_1, l_2, \dots, l_m$ where $m$ is the number of edges in $G$.

Now note that every location $l_1, \dots, l_m$ has exactly two active clocks.
We create a timed automaton which is acyclic.
We create the timed automaton in a way such that among the locations $l_1, l_2, \dots, l_m$,
there will be exactly one location in the automaton that will have no incoming edge and
one location without any outgoing edge.
We create three additional locations $l_0$, $l_{m+1}$ and $l_{m+2}$.
We add an edge from location $l_0$ to the location without
any incoming edge and similarly an edge from the location among $l_1, l_2, \dots, l_m$
without any outgoing edge to location $l_{m+1}$
and an edge from location $l_{m+1}$ to location $l_{m+2}$.
If the active clocks of $l_1$ are $x_i$ and $x_j$,
then we draw an edge from $l_0$ to $l_1$ with the constraint $x_i \le 1$ and reset $x_j$.

The other edges in the timed automaton are drawn in the following manner.
We add the edge from location $l_u$ to $l_v$ (and not from $l_v$ to $l_u$)
such that doing so maintains the acyclicity of the timed automaton.
An edge is drawn between two locations $l_u$ and $l_v$ if
$act(l_u) \cap act(l_v) \neq \emptyset$.
Note that there can be at  most one active clock common between two distinct locations $l_u$ and $l_v$.

Now we describe the constraints and the resets on the edges of the timed automaton.
Consider two locations $l_u$ and $l_v$ connected by an edge from $l_u$ to $l_v$ and let the active clocks of $l_u$ be $x_i$ and $x_j$ while the active clocks of $l_v$ be $x_j$ and $x_t$.
On this edge we add a constraint $x_i \le k$, where $k \in \mathbb{N}$ and reset $x_t$.
The value of $k$ is chosen in the following way.
Suppose clock $x_i$ was reset on an incoming edge of $l_w$ such that there is a path from $l_w$ to $l_u$ without any reset of $x_i$.
We denote by $wt(x_i)_{w, u}$, the maximum sum of the integers 
used in the constraints over all paths from $l_w$ to $l_u$.
We assign to $k$ the value $\displaystyle{1 + \max_w(wt(x_i)_{w, u})}$,
i.e. one added to the maximum of such weights at location $l_u$ computed over all incoming paths.
This value of $k$ ensures that all the valuations that reach location $l_u$,
satisfy all the constraints that appear
on the edges reachable from location $l_u$.
This leads to the fact that none of the locations are further split.

If the active clocks at some $l_i$, $1 \le i \le m$ be $x_r$ and $x_s$,
then we add an edge from $l_i$ to $l_{m+1}$
with the constraint $x_r \le k_r$ 
and an edge from $l_{m+1}$ to $l_{m+2}$
with the constraint $x_s \le k_s$
without loss of generality.
$k_r$ and $k_s$ are calculated
in the same way as $k$ as described above.

We can construct the clock graph corresponding to the timed automaton thus constructed and find
that the clock graph is isomorphic to the input graph $G$.
Thus $G$ can be coloured by renaming the clocks of the timed automaton constructed above
so that least possible number of clocks are used after the renaming the clocks.
\qed

\begin{figure}
\centering
\vspace{-10pt}
\includegraphics[width=0.85\textwidth]{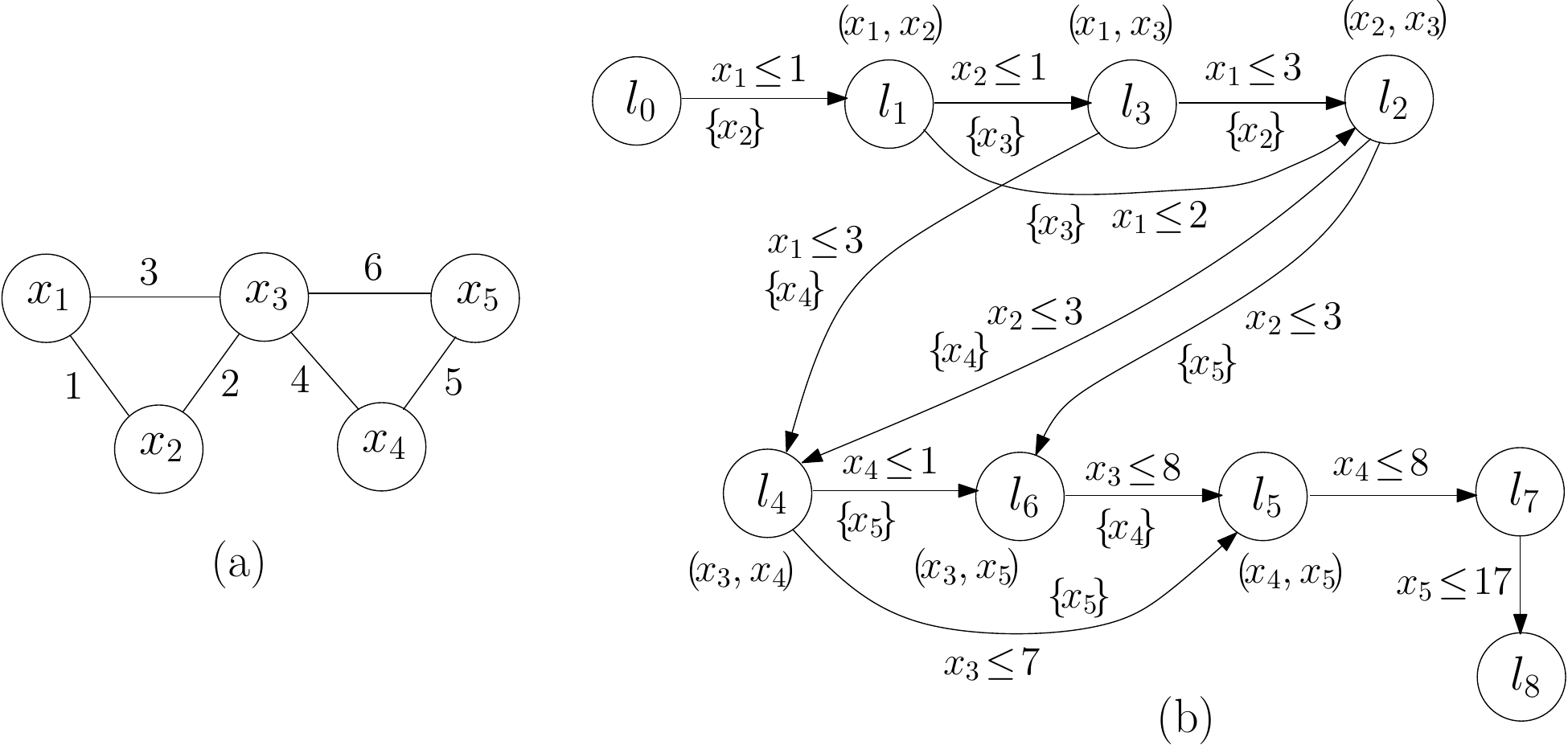}
\caption{\label{fig-reduceNP}(a) a graph $G$ (b) TA with clock graph same as $G$}
\vspace{-10pt}
\end{figure}
Figure \ref{fig-reduceNP}(a) shows a graph $G$ and \ref{fig-reduceNP}(b) shows the TA constructed
following the procedure described above.
Note that there are 6 edges in $G$ corresponding to which
there are 6 locations in the TA which are $l_1$ to $l_6$.
There are three additional locations $l_0$, $l_7$ and $l_8$.
Corresponding to every location in $l_1$ to $l_6$, its active clocks are written inside parentheses. 
For example, the active clocks of $l_1$ are $x_1$ and $x_2$ denoting that the edge in $G$ corresponding to
location $l_1$ connects vertices $x_1$ and $x_2$ in $G$.

We look at the complexity of the operations in this stage.
The sequence of computation of active clocks converges within $n$ iterations and every iteration runs in time $O(\lvert E \lvert)$, 
where there are $n$ locations and $\lvert E \lvert$ edges respectively in the timed automaton after the first three stages.
This is due to the fact that in iteration $i$, for some location $l$,
its active clocks are updated so as to include those active clocks of the locations $l'$
such that there exists a path of length at most $i$ between $l$ and $l'$
and these clocks are not reset along this path.
In each iteration, each edge is traversed once for updating the set of active clocks of the locations.
Thus the complexity of finding active clocks is $O(n \times \lvert E \lvert)$.
Partitioning the active clocks of each of the locations too requires traversing the zone graph and checking the clock relations from the DBM of the zones. This can be done in time equal to the order of the size of the zone graph times the size of DBM which is in EXPTIME.

Finally, determining the chromatic number of a graph is possible in time exponential in the number of the vertices of the graph \cite{Law1,Epp1}. Since the number of locations after the splitting operation in stage 2 is exponential in the number of clocks in $A$, renaming the clocks using the clock graph runs in time doubly exponential in the number of the clocks of the original timed automaton $A$.
Thus we have the following theorem.
\begin{theorem}
The stages mentioned above run in 2-EXPTIME.
\end{theorem}

In the presence of an invariant condition, considering an edge $l \to{g, a, R} l'$, a zone $Z'$ of $l'$ is initially created 
from zone $Z$ of $l$ such that $Z'=(Z \cap g)_{[R \leftarrow \overline{0}]}\!\!\uparrow \cap I(l')$, if $(Z \cap g)_{[R \leftarrow \overline{0}]}\!\!\uparrow \cap I(l') \neq \emptyset$, where $I(l')$ is the invariant on location $l'$.
The edge can be removed from the timed automaton if for all zones $Z$ of $l$, $(Z \cap g)_{[R \leftarrow \overline{0}]}\uparrow \cap I(l') = \emptyset$. 
The invariant condition on location $l'$ can be entirely removed if for every incoming edge $l \to{g, a, R} l'$ and for each zone $Z$ of $l$, $(Z \cap g)_{[R \leftarrow \overline{0}]}\uparrow \cap I(l') = (Z \cap g)_{[R \leftarrow \overline{0}]}\uparrow$ holds.
Considering the clock relations in the zones of a location, an elementary constraint in the invariant too is removed if it is implied by the rest of the elementary constraints in the invariant.
In stage 4, for computing the set of active clocks $act(l)$ in location $l$, $clk(l)$ is defined as the 
union of the set of clocks appearing in the constraints on the outgoing edges from $l$ and the set of clocks appearing in $I(l)$.

\textbf{Proof of Minimality of clocks:}
Let $A_4$ be the TA obtained from a TA $A$ through the four stages described earlier.
We show that all the hyperplanes induced by the constraints appearing in the
timed automata $A_4$ and appearing in its zone graph are necessary
for defining the behaviour of the timed automata.
Further, we can show that these constraints in $A_4$ are expressed using the minimum possible number of clocks.

The following lemma states that every hyperplane in
the zone graph of $A_4$ is needed for preserving bisimulation.
\begin{lemma} \label{lem-hyperplane}
In the zone graph of $A_4$, every hyperplane (modulo clock renaming) 
that bounds a zone from above is necessary.
\end{lemma}
\proof
Every constraint appearing in the timed automaton $A_4$ induces a hyperplane
in the zone graph of the automaton.
We need to show that each of these hyperplanes defines the behaviour of timed automaton;
stated otherwise, removal of any of these hyperplanes will alter the behaviour of
the initial state of the timed automaton
and with this modification, the initial state of the automaton
will no longer be timed bisimilar to the initial state of 
any timed automaton timed bisimilar to $A_4$.

In stage 3, we consider constraints on multiple outgoing edges from a location
such that each of these edges is labelled with the same action.
In this stage, a zone is merged with its immediate successor zone if and only if
by doing so, the initial state remains timed bisimilar to the
one before merging the zones.

Thus in the zone graph of $A_4$, if a hyperplane exists, it implies that
it could not be removed due to merging of zones in stage 3
since by doing so, the automaton does not remain bisimilar any more.
Hence the hyperplane is necessary for preserving timed bisimulation.

Since we argue about any arbitrary hyperplane in the zone graph of $A_4$,
we can say that every hyperplane in the zone graph of $A_4$ is necessary for preserving
bisimulation property of $A_4$.
\qed
\begin{lemma} \label{lem-plane}
In the TA $A_4$, for each location $l$ and clock $x \in act(l)$, there exists a constraint involving $x$
on some outgoing edge of $l$ or on the outgoing edge of another location $l'$ reachable from $l$ such that
there is at least one path from $l$ to $l'$ without any intervening reset of $x$.
This constraint induces a hyperplane in the zone graph of $A_4$ which
if removed so as to not bound a zone in the zone graph of the timed automaton,
does not preserve timed bisimulation property of the timed automaton $A$.
\end{lemma}
\proof
We consider a location $l \in L$ and suppose $x \in act(l)$.
By the definition of active clocks,
there is a constraint involving clock $x$ on the outgoing edge of $l$
or on the outgoing edge of a location $l'$
such that there is a path from $l$ to $l'$
on which clock $x$ is not reset.

Now a constraint $x \bowtie k$ involving clock $x$ induces a hyperplane
in the zone graph of automaton $A_4$ that is of the form $x = k$.
Since from lemma \ref{lem-hyperplane}, every hyperplane appearing in the zone
graph of $A_4$ is necessary for preserving bisimulation,
this hyperplane too cannot be removed.

Since we consider an arbitrary location $l$ and some active clock $x$
of $l$, the lemma thus holds true for all active clocks in each location
of the timed automaton $A_4$.
\qed
\begin{definition}
\textbf{Minimal bisimilar TA}\index{minimal bisimilar TA}: For a given timed automaton $B$, a minimal bisimilar TA $B_1$ is one that is timed bisimilar to $B$ and has the minimum number of clocks possible.
\end{definition}
\begin{fact}
For every TA $B$, there exists a minimal bisimilar TA.
\end{fact}
One can show that the clocks of a TA that is minimal bisimilar to $A_4$ can replace the clocks of $A_4$ which gives us the following lemma.
\begin{lemma} \label{lem-minclocks}
The timed automaton $A_4$ has the same number of clocks as a minimal bisimilar TA for $A$.
\end{lemma}
%
\proof
We prove the lemma by showing that the clocks of $A_4$ can be replaced with the clocks of a minimal bisimilar TA that uses say $\lvert C_B\lvert$ clocks.
Let $A_4$ has $\lvert C_{A_4}\lvert$ clocks.
We know that if $B$ is a minimal TA, we have $\lvert C_B\lvert  \le \lvert C_{A_4}\lvert $.
Replacing clocks of a timed automaton
with another set of clocks is essentially \emph{renaming} of the clocks.
Since due to the operations in stage 4, any renaming cannot reduce the number of clocks in $A_4$
and if the clocks of $A_4$ can be replaced with the clocks of $B$,
we have $\lvert C_{A_4}\lvert  \le \lvert C_B\lvert $, thus giving $\lvert C_{A_4}\lvert  = \lvert C_B\lvert $.

We consider a minimal TA $B_1$ timed bisimilar to the given TA $A$ and apply the operations in the four stages on it to produce a TA $B$. The application of these stages on the minimal TA does not add to the number of clocks and since $B_1$ is already minimal, $B$ has the same number of clocks as $B_1$.
Due to the transformation from $B_1$ to $B$, it is
ensured that in $B$'s zone graph, for each zone in every location,
there exists a hyperplane that fully bounds it.
Similarly in the zone graph of $A_4$ too, for each zone of every location, there exists a hyperplane that fully bounds it.

We use the zone graphs of $A_4$ and $B$ for mapping the clocks of $A_4$ to the clocks of $B$.
From Lemma \ref{lem-plane}, for a location $l_{A_4}$ of $A_4$, corresponding to every clock $x \in act(l_{A_4})$,
a hyperplane of the form $x = k$ corresponding to a constraint $x \bowtie k$
exists in the zone graph of $A_4$
which cannot be removed so that the modified TA remains timed bisimilar.
Since $A_4$ and $B$ are timed bisimilar, there is a corresponding hyperplane, say $y = k'$,
induced by a constraint $y \bowtie k'$ which too cannot be removed from $B$ while preserving timed bisimulation.

In the following part of the proof,
the cp-traces we consider are the ones that start from the initial state of the zone graph and
the delays in the trace are such that it moves from an entry corner point to the exit corner point of the same zone and from an exit corner point to an entry corner point of the immediate delay successor zone.
We consider multiple finite cp-traces in the zone graph such that all the corner points in the zone graph are traversed by some cp-trace at least once.

Now we consider a part of a cp-trace in the zone graph of $A_4$ from the initial state to an 
exit corner point $(l_{A_4}, v)$ of zone $Z_{A_4}$.
Suppose the hyperplane $x = k$ bounds the zone $Z_{A_4}$
and the constraint $x \bowtie k$ induces the hyperplane and it cannot be removed
while preserving timed bisimulation.
Since $A_4$ and $B$ are timed bisimilar,
there exists a state $(l_B, v')$ in the zone graph of $B$ that is timed bisimilar to $(l_{A_4}, v)$
and $(l_B, v')$ is an exit corner point of some zone $Z_B$.
Continuing this way, we consider \emph{all} the different exit corner points in $Z_{A_4}$ and
the corresponding bisimilar corner points in $Z_B$ to find out the clock involved in the hyperplane bounding the zone $Z_{A_4}$, that is $x$ and the clock in the corresponding hyperplane bounding $Z_B$ that is say $y$ and thus we can replace clock $x$ in $A_4$ with clock $y$ of $B$. This is repeated until all the occurrences of all the clocks in $A_4$ are replaced with clocks in $B$.

As mentioned above, here $x \in act(l_{A_4})$ and $y \in act(l_B)$.
Since an active clock of a location cannot be replaced with another active clock
of the same location, every active clock in every location in $A_4$, 
can be mapped uniquely to a clock in $B$
which can replace the clock in $A_4$.
Stated otherwise, if $x$ could be replaced with any of the two clocks
$y_1, y_2 \in act(l_B)$,
then in $B$ itself, $y_1$ could have replaced $y_2$ or vice versa.
Continuing this way, we have a mapping for each clock in $A_4$ and
the clocks associated to the constraints in
all the edges in $A_4$ are replaced with clocks in $B$.

Now since the clocks of $A_4$ have been replaced with the clocks of $B$
and from the operations of stage 4,
any renaming of clocks of $A_4$ requires at least as many clocks as in $A_4$,
we have $\lvert C_{A_4} \lvert \le  \lvert C_B \lvert$.
On the other hand, since $B$ is a minimal bisimilar TA, we have $\lvert C_B \lvert \le C_{A_4} \lvert$. 
Thus we have $\lvert C_{A_4} \lvert = \lvert C_B \lvert$, i.e. $A_4$ and $B$ have the same number of clocks. 
Stated otherwise, $A_4$ has the least number of clocks among the TAs that are timed bisimilar to the original timed automaton $A$.
\qed

\begin{theorem}\index{clock reduction!preserving timed bisimulation}
There exists an algorithm to construct a TA $A_4$ that is timed bisimilar to a given TA $A$
such that among all the timed automata that are timed bisimilar to $A$, $A_4$ has the minimum number of clocks.
Further, the algorithm runs in time that is doubly exponential in the number of clocks of $A$.
\end{theorem}

\section{Conclusion}\label{sec-conc}
In this paper, we have described an algorithm, which given a timed automaton $A$, produces another timed automaton $A_4$ with the smallest number of clocks that is timed bisimilar to $A$.
Since $A_4$ is timed bisimilar to $A$, it also accepts the same timed language as $A$.
The problem appears in a different form
in \cite{LLW1} where it was left open.
In \cite{LLW1}, this problem has been described as 
given a timed automaton $A$ and a set $C$ of clocks,
whether there exists a timed automaton that is
timed bisimilar to the original automaton $A$ and has $|C|$ clocks.
It is easy to see that a solution to our problem solves
the problem in \cite{LLW1} and vice versa.

For reducing the number of clocks of the timed automaton, we rely on a semantic representation of the timed automaton rather than its syntactic form as in \cite{DY1}. This helps us to reason about the behaviour of the timed automaton more effectively.
We could just as well have used a region graph construction,
the asymptotic complexity would remain the same.
However, the zone graph we use in our approach is usually much smaller in size than the region graph.
Besides, the size of the zone graph is independent of the constants used in the timed automaton.
There is an exponential increase in the number of locations while producing the TA with the minimal number of clocks.
However, there is no addition of new valuations to the underlying state space of the timed automaton in stage 2,
since the splitting of a location $l$,
described in stage 2, involves distributing the zones of $l$ across the locations $l$ is split into.

We note that pre-stability of the zones with respect to transitions is essential for the operations in stage 2.
Notice that regions are inherently pre-stable \cite{Sri1} and therefore the region graph construction would work here as well
though the number of locations obtained through splitting would be larger than in the case of zone graphs.

Our algorithm has a 2-EXPTIME upper bound,
however, it needs to be used only once to reduce the number of clocks.
Reachability analysis and often model checking
for timed automata like TCTL model checking are PSPACE-complete \cite{AD1,AL1,ACD1}.
Since reachability analysis and model checking the timed automata
are usually performed several times to verify different properties,
we may consider first reducing the number of clocks before doing the model checking.
With fewer clocks, the DBM representation of a zone will also be smaller in size
which may as well contribute towards reducing the computation time for model checking.
\bibliographystyle{plain}
\bibliography{ClockReduction}

\begin{thebibliography}{10}

\bibitem{LAKJ1}
L.~Aceto, A.~Ing$\acute{o}$lfsd$\acute{o}$ttir, K.G. Larsen, and J.~Srba.
\newblock {\em Reactive Systems: Modelling, Specification and Verification}.
\newblock Cambridge University Press, 2007.

\bibitem{AL1}
L.~Aceto and F.~Laroussinie.
\newblock Is your model checker on time? on the complexity of model checking
  for timed modal logics.
\newblock In {\em Proceedings of MFCS}, pages 125--136, London, UK, 1999.
  Springer-Verlag.

\bibitem{ACD1}
R.~Alur, C.~Courcoubetis, and D.~Dill.
\newblock Model-checking in dense real-time.
\newblock {\em Inf. Comput.}, 104:2--34, May 1993.

\bibitem{AD1}
R.~Alur and D.L. Dill.
\newblock A theory of timed automata.
\newblock {\em Theoretical Computer Science}, 126:183--235, 1994.

\bibitem{BBFL1}
G.~Behrmann, P.~Bouyer, E.~Fleury, and K.~G. Larsen.
\newblock Static guard analysis in timed automata verification.
\newblock In {\em TACAS}, pages 254--270, 2003.

\bibitem{BBLP1}
G.~Behrmann, P.~Bouyer, K.~G. Larsen, and R.~Pelanek.
\newblock Lower and upper bounds in zone-based abstractions of timed automata.
\newblock {\em International Journal on Software Tools for Technology
  Transfer}, 8:204--215, 2006.

\bibitem{BFT1}
Alberto Bemporad, Komei Fukuda, and Fabio~Danilo Torrisi.
\newblock Convexity recognition of the union of polyhedra.
\newblock {\em Computational Geometry}, 18(3):141--154, 2001.

\bibitem{BM1}
B.~Berthomieu and M.~Menasche.
\newblock An enumerative approach for analyzing time petri nets.
\newblock In {\em IFIP Congress}, pages 41--46, 1983.

\bibitem{Cer1}
K.~Cerans.
\newblock Decidability of bisimulation equivalences for parallel timer
  processes.
\newblock In {\em Proceedings of CAV}, volume 663, pages 302--315.
  Springer-Verlag, 1992.

\bibitem{DT1}
C.~Daws and S.~Tripakis.
\newblock Model checking of real-time reachability properties using
  abstractions.
\newblock In {\em TACAS}, pages 313--329, 1998.

\bibitem{DY1}
C.~Daws and S.~Yovine.
\newblock Reducing the number of clock variables of timed automata.
\newblock In {\em Proc. RTSS'96, 73-81, IEEE}, pages 73--81. IEEE Computer
  Society Press, 1996.

\bibitem{Dill1}
D.~L. Dill.
\newblock Timing assumptions and verification of finite-state concurrent
  systems.
\newblock In {\em Automatic Verification Methods for Finite State Systems},
  pages 197--212, 1989.

\bibitem{Epp1}
D.~Eppstein.
\newblock Small maximal independent sets and faster exact graph coloring.
\newblock {\em J. Graph Algorithms Appl.}, 7:131--140, 2003.

\bibitem{Fin1}
O.~Finkel.
\newblock Undecidable problems about timed automata.
\newblock In {\em Proceedings of FORMATS}, pages 187--199. Springer-Verlag,
  2006.

\bibitem{GKNA1}
S.~Guha, S.~N. Krishna, C.~Narayan, and S.~Arun-Kumar.
\newblock A unifying approach to decide relations for timed automata and their
  game characterization.
\newblock In {\em Express/SOS}, pages 47--62, 2013.

\bibitem{HNSY1}
T.~A. Henzinger, X.~Nicollin, J.~Sifakis, and S.~Yovine.
\newblock Symbolic model checking for real-time systems.
\newblock {\em Information and Computation}, 111:394--406, 1992.

\bibitem{LLW1}
F.~Laroussinie, K.~G. Larsen, and C.~Weise.
\newblock From timed automata to logic -- and back.
\newblock In {\em Proceedings of MFCS}, pages 529--539, 1995.

\bibitem{LH1}
F.~Laroussinie and Ph. Schnoebelen.
\newblock The state explosion problem from trace to bisimulation equivalence.
\newblock In {\em FoSSaCS}, pages 192--207, 2000.

\bibitem{Law1}
E.~L. Lawler.
\newblock A note on the complexity of the chromatic number problem.
\newblock {\em Inf. Process. Lett.}, 5:66--67, 1976.

\bibitem{Sri1}
Balaguru Srivathsan.
\newblock {\em Abstractions for timed automata.}
\newblock PhD thesis, LaBRI, University of Bordeaux, 2012.

\bibitem{Trip1}
S.~Tripakis.
\newblock Folk theorems on the determinization and minimization of timed
  automata.
\newblock {\em Inf. Process. Lett.}, 99:222--226, September 2006.

\bibitem{TY1}
S.~Tripakis and S.~Yovine.
\newblock Analysis of timed systems using time-abstracting bisimulations.
\newblock {\em Formal Methods in System Design}, 18:25--68, 2001.

\bibitem{WL1}
C.~Weise and D.~Lenzkes.
\newblock Efficient scaling-invariant checking of timed bisimulation.
\newblock In {\em STACS}, pages 177--188, 1997.

\end{thebibliography}
\end{document}